\documentclass[12pt]{article}
\usepackage[left=.8in,top=1in,right=.8in,bottom=.5in,head=.1in,nofoot]{geometry}
\usepackage{multirow}
\usepackage{setspace}
\usepackage{rotating}

\usepackage{booktabs}
\usepackage{hyperref}
\renewcommand{\arraystretch}{0.85}

\usepackage{amsmath, amssymb, amsfonts, bm}
\usepackage{latexsym}
\usepackage{bbm}
\usepackage{titlesec}        % (optional) style sections
\usepackage{fancyhdr}        % (optional) headers/footers

\usepackage{amsthm}

\newtheorem{theorem}{Theorem}

\newtheorem{lemma}{Lemma}
\newtheorem{corollary}{Corollary}

\theoremstyle{definition}

\theoremstyle{remark}

\usepackage{graphicx}
\usepackage{subcaption}
\usepackage{float}
\usepackage{booktabs}
\usepackage{siunitx}

\usepackage{algorithm}
\usepackage{algpseudocode}

\usepackage{xcolor}

\usepackage{natbib}
\bibpunct{(}{)}{;}{a}{}{,}

\usepackage{etoolbox}
\apptocmd{\sloppy}{\hbadness 10000\relax}{}{}

\usepackage{titlesec}

\titleformat*{\section}{\bf\large}

\usepackage{authblk}

\usepackage{enumitem}

\title{Finite-Sample Valid Rank Inference for a Broad Class of Statistical and Machine Learning Models} 
\title{Finite-Sample Valid Rank Confidence Sets for a Broad Class of Statistical and Machine Learning Models}
\author{Onrina Chandra, and Min-ge Xie\footnote{Onrina Chandra (email: oc152@stat.rutgers.edu) is a graduate student and Min-ge Xie (email: mxie@stat.rutgers.edu) are  Professors, Department of Statistics, Rutgers, The State University of New Jersey, Piscataway, NJ 08854.  The research is supported in part by NSF grants DMS2015373, DMS2027855, DMS2311064 and DMS-2319260.}
}

\date{}

\begin{document}

% Turn off paragraph indentation
\setlength{\parindent}{0pt}
% (Optional) Add small spacing between paragraphs for readability
\setlength{\parskip}{6pt}

\thispagestyle{empty}

\setcounter{page}{1}
 \maketitle
\vspace{-10mm}
\begin{abstract}
\noindent
Ranking populations such as institutions based on certain characteristics is often of interest, and these ranks are typically estimated using samples drawn from the populations. Due to sample randomness, it is important to quantify the uncertainty associated with the estimated ranks. This becomes crucial when latent characteristics are poorly separated and where many rank estimates may be incorrectly ordered. Understanding uncertainty can help quantify and mitigate these issues and provide a fuller picture. However, this task is especially challenging because the rank parameters are discrete and the central limit theorem does not apply to the rank estimates. 
In this article, we propose a Repro Samples Method to address this nontrivial %irregular 
inference problem by developing a confidence set for the true, unobserved population ranks. This method provides finite-sample coverage guarantees and is broadly applicable to ranking problems. 
The effectiveness of the method is illustrated and compared with several published large sample ranking approaches using simulation studies and real data examples involving samples both from traditional statistical models and modern data science algorithms.
\end{abstract}
\noindent 
Key words: Inference on discrete parameter space; Finite-sample performance guarantee;
% Discrete parameter space; irregular inference problem; 
Repro sample method; Latent model; Rank of performance. 
% \vspace{3mm}

%\newpage
\section{Introduction}

The ranking performance of institutions such as universities, hospitals or sports teams, plays a crucial role in shaping decisions across many areas. Prospective students choose colleges based on league tables, patients rely on hospital ratings when seeking care and sponsors follow conference standings to gauge sport teams performance. These rankings are almost always derived from sampled data, imperfect measurements, or subjective evaluations, which introduce variability and potential biases. A university’s placement in a league table, for example, may shift simply because of small fluctuations in survey responses. Similarly, a hospital’s star rating can move up or down if a few outcome metrics change. Ignoring this uncertainty can mislead decisions; students may choose a school based on noise, or resources may go to a hospital whose top rank is unstable.
Therefore, it is crucial to develop statistical tools that not only estimate rank but also quantify its uncertainty. Confidence intervals help distinguish real differences from random variation, promoting transparent, evidence-based decisions and avoiding the false certainty of exact ranks.

 In this paper, we aim to rank $K$  populations $\mathcal{P}_1,\mathcal{P}_2,...\mathcal{P}_K$ that are defined through a characteristic described by a set of unknown (numeric or non-numeric) feature values $\bm{\eta} = (\eta_1,\eta_2,...,\eta_K),^\top$ where $\eta_k$ is a set of features associated with the $k^{th}$ population $\mathcal{P}_k$, for $k=1,2,..,K$, and $\bm{\eta} \in \Omega$, an arbitrary feature space. More specifically, we assume the populations are ranked based on a characteristic parameter $\bm{\theta} = (\theta_1, \theta_2, \ldots, \theta_K)^\top$ of $K$ elements, where $\bm{\theta} = \zeta(\bm{\eta})
$ is a function of features  $\bm{\eta}$ for a mapping function $\zeta(\cdot)$ from $\Omega \rightarrow \Theta  \subseteq \mathbb{R}^K$. Our target is the ranks of the $K$ populations, denoted by  $\bm{R}$, is determined by the ranking order of the $K$ elements~of~$\bm\theta$: 
\begin{equation}
\label{eq:rank and theta}
\bm{R} = (r_1 , \ldots, r_K )=\mathcal{S}(\bm{\theta}) =\mathcal{S}(\theta_1, \ldots, \theta_{K}) \in [K]^{{K}}
% {\cal I},
\end{equation}
where the rank of the $k^{th}$ population is defined as 
$r_k  = \sum_{1 \leq i \leq K, i \not = k} {\bf 1} (\theta_i \leq \theta_k)$ and  $\mathcal{S}(\cdot)$ is the corresponding mapping function from $\theta  \to [K]^{K}$, where the notation $[K]$ denotes the set of the first $K$ positive integers, $\{1,2,\dots,K\}$ throughout the paper. 

In practice, the $\theta_k$ values are unknown, 
but we have sample data, say ${\cal D}$, collected from the $K$ populations. The population ranks $\bm R$ are often estimated by replacing $\theta_k$'s in (\ref{eq:rank and theta})  with their estimate 
$\hat{\theta}_k=\hat{\theta}_k({\cal D})$'s  using the sample data and consequently $\widehat{\bm{R}} 
= (\hat r_1 , \ldots, \hat r_K )= \mathcal{S}(\widehat{\bm{\theta}}) $,  where 
$\hat{r}_k  = \sum_{1 \leq i \leq K, i \not = k} {\bf 1} (\hat{\theta}_i \leq \hat{\theta}_k)$ and $\widehat{\bm{\theta}} = (\hat \theta_1, \ldots, \hat \theta_{K})$. 
However, the uncertainty and error of estimators can lead to discrepancies between the ordering derived from  $\hat{\theta}_k$’s and the actual ordering dictated by the $\theta_k$ values, particularly when some of the underlying $\theta_k$ values are closely clustered. This potential for misordering highlights the importance of developing methods to provide confidence sets instead of just point estimates of ranks. For clarity and mathematics rigor, throughout the paper, we
use ${\cal D}^{\mathrm{obs}}$ to denote the observed (nonrandom) sample of ${\cal D}$. 
We also use $\bm{\eta}^{(0)}$ and $\bm{\theta}^{(0)}$ to denote the true values of $\bm{\eta}$ and $\bm{\theta}$. We assume throughout the paper the true scores $\theta^{(0)}_1,\dots,\theta^{(0)}_K$ are distinct although they can be very close to each other. Thus, we seek to infer about the parameter of interest which is the true population rank $\bm{R}^{(0)}$ using the observed sample data ${\cal D}^{\mathrm{obs}}.$ 
\begin{equation}
\label{eq:rank and theta true}
\bm{R}^{(0)} = (r^{(0)}_1 , \ldots, r^{(0)}_K )=\mathcal{S}(\bm{\theta}^{(0)}) =\mathcal{S}(\theta^{(0)}_1, \ldots, \theta^{(0)}_{{K}})
\end{equation}

\subsection{Challenges in Rank Inference}

Rank inference problems differ from most statistical inference problems because the parameter space of interest is discrete rather than continuous, which poses significant challenges for applying standard inferential tools.
For instance, inference tools based on large-sample methods, such as the Central Limit Theorem, are generally not applicable because the regularity conditions required for the theorem, do not hold for point estimators of ranks. Frequentist approaches, such as bootstrap methods based on rank estimators, falter in this setting because the limiting distributions of discrete rank estimators are often unknown, and the bootstrap Central Limit Theorem does not apply. As noted in \cite{hall2010}, even the limiting distribution of $\max_k \hat{\theta}_k$ is elusive, complicating the development of valid inference procedures.
 Bayesian methods, though capable of generating credible sets for ranks, also encounter fundamental difficulties. Particularly, credible sets obtained using
the posteriors distributions perform poorly in terms of covering the true rank in repeated runs because the Bernstein–von Mises theorem does not extend to discrete parameter spaces. Moreover, the performance of the credible sets depends critically on the choice of priors, for which there is no commonly agreed-upon choice. Again, when a Bayesian method enforces continuous prior on the latent parameter space $\bm{\Theta}$, it breaks any potential ties that may exist between the true population parameters. Our work addresses this critical gap by sidestepping these issues and developing a novel procedure to produce finite-samples valid confidence sets for the true population ranks.  
% based on finite samples.
% , thereby enabling more reliable decision-making in a wide range of applications.

\subsection{Main goal of the paper}

Given observed data $\mathcal D^{\mathrm{obs}}$, we aim to construct a set 
$\tilde{\Gamma}_{\mathcal{V}_\alpha}(\mathcal D^{\mathrm{obs}})\subset [K]^K$ of rank vectors that satisfies the finite-sample coverage constraint $\mathbb P \!\big(\bm R^{(0)} \in \tilde{\Gamma}_{\mathcal{V}_\alpha}(\mathcal D)\big) \ge 1-\alpha,$ for a prescribed confidence level $\alpha \in (0,1)$. This formulation parallels the construction of finite-sample valid prediction sets for 
discrete outcomes, but the inferential target here is an ordering $S(\bm\theta)$ rather than 
a continuous parameter. Within the Repro-Samples framework, our method extends the notion of exact coverage from $\bm \Theta$ to its induced rank space.

\subsection{Existing Approaches to Rank Inference}

The task of quantifying uncertainty in the ordering of latent parameters $\theta,\dots,\theta_K$ has progressed through a network of interconnected breakthroughs. 
Early work by \cite{goldstein} recognized that posterior uncertainty in hierarchical models should induce uncertainty in ranks. 
They aimed to critically examine the statistical challenges of “league tables” for comparing institutional performance using hierarchical models with shrunken estimates of institution-specific effects $\theta_k$, applying normal priors and Gibbs-sampler-based intervals. 
More recently, to provide joint credible intervals for the ranks under a Bayesian framework, \cite{datta2024credibledistributionsoverallranking} compared an unstructured flat prior $\pi(\theta)$ and the Fay–Herriot prior on $\theta_k$, using MCMC to produce full posterior rank‐distribution matrices.  
\cite{gu2023invidious} recast ranking as a compound‐decision problem by modeling $\theta_k$ through a mixing distribution $G$, estimated via nonparametric MLE, thereby bridging empirical Bayes estimation and predictive ranking.

Unlike Bayesian credible intervals, which reflect posterior belief, frequentist methods focus on producing confidence sets that achieve the claimed coverage probability. 
Early frequentist approaches recognized that naive “plug-in” ranks $\hat{r}_k = 1 + \sum_{i \neq k} \mathbf{1}\{\hat{\theta}_i < \hat{\theta}_k\}$ break down in the presence of ties or near-ties. 
\cite{xie2009} addressed this by replacing the discrete indicator with a smooth approximation. 
They assumed root-$n$ estimators $\hat{\theta}_{kn} = \theta_k + n^{-1/2}Z_{kn} + o_P(1)$, where $Z_{kn}$ converges in law to $N(0,\sigma_k^2)$, and defined a smoothed version $\hat{R}_{kn}^{\mathrm{smooth}} = 1 + \sum_{i \neq k} G_n(\hat{\theta}_{in} - \hat{\theta}_{kn})$ of the plug-in ranks, where $G_n$ approximates the step function. 
By tuning the smoothing bandwidth, they proved that $\hat{R}_{kn}^{\mathrm{smooth}}$ is consistent for $R_k$ and developed a specialized bootstrap that remains valid even with near-ties among the $\theta_k^{(0)}$. 
Simultaneously, \cite{hall2010} demonstrated the failure of the naive $n$-out-of-$n$ bootstrap for discrete ranks (where $n$ is the average sample size across the $K$ populations) and advocated an $m$-out-of-$n$, $m<n$ resampling scheme to restore consistency.

A parallel strand developed exact finite‐sample constructions via optimization and likelihood. 
\cite{liu2022lagrangian} reformulated rank‐confidence‐interval construction as an integer program, first forming normal confidence intervals for each contrast 
$\theta_k - \theta_i \in \hat{\theta}_k - \hat{\theta}_i \pm z_{\alpha/2}\sqrt{\sigma_k^2/n_k + \sigma_i^2/n_i}$, 
then applying a Lagrangian relaxation to find the minimal and maximal ranks consistent with all these intervals. 
\cite{almohamad2022simultaneous} proposed simultaneous $1-\alpha$ confidence intervals for the true ranks $r_k$ in independent Gaussian samples $Y_k \sim N(\theta_k,\sigma_k^2)$ using Tukey’s honest significant difference method. 
which could be conservative when the $\theta_i$’s are close together. 
Addressing a selection problem, \cite{andrews2019inference} studied the “winner’s curse” that arises when one first selects a parameter $\hat{a}$ by optimizing over a finite set and then conducts inference on its effect $\theta(\hat{a})$. 
They derived conditional truncated‐normal confidence intervals for the “winner” $\hat{\theta} = \arg\max_{\theta \in \Theta} X(\theta)$ and proposed intervals  achieving exact conditional coverage $1-\alpha$. 

A complementary multiple‐testing approach to rank inference focused squarely on the family of pairwise hypotheses $H_{ik}\!:\theta_k \le \theta_i$. 
\cite{holm2013confidence} constructed intervals using a step‐down procedure at level $\alpha/(K-1)$, counting rejections $N_k^-$ and $N_k^+$ on each side and setting 
$\mathrm{CI}_k = [1+N_k^-,\,K-N_k^+]$ with exact family‐wise error control.  
\cite{klein2020joint} introduced joint confidence sets $\{L_k,U_k\}$ for each $\theta_k$ via Bonferroni or exact methods, defining
$r_k \in \{\,|\{i:U_i<L_k\}|+1,\dots,|\{i:L_i\le U_k\}|\,\}$, thereby obtaining valid rank sets without resampling. 
\cite{mogstad2024inference} sharpened this approach by constructing uniform simultaneous bands for all contrasts $\theta_i-\theta_k$ via the maximum of studentized statistics obtaining rank sets under directional FWER control.  
Specializing to categorical data and accounting for dependence, \cite{bazylik2021finite} employed UMPU conditional‐binomial tests for $\theta_j \le \theta_k$, combined with Holm/Bonferroni adjustments, to deliver finite‐sample exact rank intervals. 
However, both the number of tests and the complexity of the Holm procedure scale quadratically, which can become prohibitive when there are many populations.

From an algorithmic standpoint, methods in the machine‐learning literature reframed ranking as a predictive task. 
\cite{furnkranz2003pairwise} introduced pairwise‐preference learning, defining a ranking function that maps instances to total orders over a set of labels.
\cite{negahban2012iterative} proposed Rank Centrality, an iterative rank‐aggregation algorithm that estimates scores from pairwise comparisons, with finite‐sample error bounds scaling as $\mathcal{O}(n\log n)$. 
In the Plackett–Luce setting, \cite{azari2013generalized} proposed a generalized method‐of‐moments estimator, breaking full rankings into moment equations and solving for latent utilities $\bm{\theta}$.  
\cite{chen2019spectral} refined Bradley–Terry estimation by combining a spectral initialization with coordinate‐wise maximum‐likelihood updates to achieve minimax error rates under suitable separation conditions. 
More recent work has established rigorous inference in sparse‐comparison and multiway models. 
\cite{han2020asymptotic} proved that, on an $n$‐vertex Erdős–Rényi graph with edge probability $p_n\gtrsim(\log n)^3/n$, the Bradley–Terry MLE $\widehat{\theta}_i$ satisfies 
$\sqrt{n p_n}\bigl(\widehat{\theta}_i/\theta_i - 1\bigr)\xrightarrow{d}N(0,\Sigma_{ii}^{-1})$, 
enabling rank intervals by counting significant log‐score differences. 
\cite{han2023general} generalized this to arbitrary sparse networks under parametric links, establishing uniform consistency $\widehat{\bm{\theta}}$ and a componentwise CLT. 
\cite{chen2021optimalrankingpairwisecomparisons} showed that the optimal error rate for recovering a full ranking under the Bradley–Terry–Luce model exhibits a sharp threshold—exponential decay in one regime and polynomial in another. 
Building on this, \cite{gao2021uncertainty} derived precise finite‐sample approximations for both MLE and spectral estimators even under sparsity, establishing central‐limit results and confidence intervals for each rank. 
\cite{han2025unifiedanalysislikelihoodbasedestimators} recently unified full, marginal, and quasi‐MLE estimators in the Plackett–Luce model on hypergraphs under a rapid‐expansion condition, providing smoothed‐rank estimators with bootstrap corrections that cover near‐ties. 
\cite{fan2024ranking} addressed multiway Plackett–Luce comparisons, observing only the top choice in each $M$-length subset among $K$ populations or items, and applied a Gaussian multiplier bootstrap on all pairwise contrasts to construct valid rank intervals. 
Existing methodologies for rank inference can be organized into four broad categories:  
(i) \emph{Asymptotic frequentist methods}, which rely on root-$n$ approximations or CLTs %(e.g., \cite{xie2009,hall2010,han2023general,chen2019spectral});  
(ii) \emph{Finite sample based multiple testing procedures} %(e.g., %\cite{holm2013confidence,bazylik2021finite,mogstad2024inference}), 
which provide exact control but can be conservative and computationally demanding;  
(iii) \emph{Likelihood or optimization‐based finite‐sample procedures}, %(e.g., \cite{liu2022lagrangian,almohamad2022simultaneous,andrews2019inference}),
which guarantee finite‐sample validity but often depend on strong model assumptions and  
(iv) \emph{Bayesian credible‐set methods} %(e.g., \cite{goldstein,datta2024credibledistributionsoverallranking,gu2023invidious}),
which quantify posterior belief rather than frequentist coverage. Within this taxonomy, our Repro‐Samples approach belongs to the finite‐sample frequentist class but differs fundamentally from optimization‐ or bootstrap‐based methods. 
By explicitly reproducing the noise generating process rather than resampling the data or estimating asymptotic distributions, our procedure constructs confidence sets for ranks that achieve exact finite‐sample coverage under minimal assumptions on the data‐generating mechanism.  
Our method bypasses reliance on point estimators or asymptotic approximations and remains effective even when the parameters are closely spaced, something which previous approached fell short of.
Conceptually, it provides a bridge between finite‐sample coverage and rank uncertainty quantification, establishing a new class of nonasymptotic, model‐agnostic rank‐confidence methods.

\subsection{Main Contributions}

This paper makes four main contributions. 
First, it extends the Repro-Samples principle to discrete rank parameters, demonstrating that finite-sample validity can be achieved by reproducing model noise rather than resampling data or relying on asymptotic distributions. 
Second, it introduces a constraint-based construction of candidate rank sets, in which a discordance constraint limits the number of pairwise order reversals between model-implied and data-implied ranks, thereby ensuring computational tractability and interpretability. 
Third, it establishes non-asymptotic coverage guarantees and characterizes how the size of the candidate set depends on the discordance budget and the number of repro samples. 
Finally, the framework is shown to unify several ranking settings, including nonparametric quantile ranking, regression-based comparisons, and partial rankings under the Plackett–Luce model, providing both joint and marginal rank confidence sets that are validated through theoretical analysis and simulation studies.

\subsection{Sample Data and Model Setup}
 We consider a very general setup that encompasses almost all scenarios encountered in practice, where the sample data $\mathcal D$ may consist of individual‐level information and/or interaction (network) data spanning multiple institutions. 
To set notation, we assume that we observe an $n\times1$ response vector $\bm{Y}\subseteq\mathcal{Y}$ with an $n\times q$ design matrix $\bm{X}\subseteq\mathcal{X}$, where $q\ge1$, and write
\[
\mathcal D = (\bm{Y}, \bm{X})
\]
The matrix $\bm{X}$ may encode covariates or features of institutions (Section~3.2) or indices linking observations to institutions (Section~3.3). 
For example, in the English Premier League (EPL) 2024 ranking application (Section~4.2), $\bm Y$ represents the vector of game scores and $\bm{X}$ records the fixed team identifiers of the competing clubs. 
Although in most examples $\mathcal Y$ and $\mathcal X$ are subsets of Euclidean spaces, this is not required; for instance, in the Plackett–Luce network model,  $\bm Y$ consists of a set of item indices together with a partial ranking outcome. 
We use a superscript “$\mathrm{obs}$’’ to denote observed (non‐random) quantities, while the corresponding random versions are denoted by $\bm{Y}$ and $\mathcal D$. For notational simplicity, we assume that $\bm{X}$ is fixed (non-random) for conditional inference; alternatively, any random components in $\bm{X}$ could possibly be absorbed into $\bm{Y}$. Our modeling assumption is deliberately minimal. 
Irrespective of whether the underlying mechanism is statistical or algorithmic, parametric or nonparametric, we require only that the random data $\mathcal D$ contain sufficient information to recover a $K$‐dimensional latent vector of characteristic parameters $\boldsymbol\theta$ through a deterministic mapping
\begin{equation}
\label{eq:inverted-general}
\boldsymbol\theta = H(\mathcal D, \bm U),
\end{equation}
where $H(\cdot)$ denotes a function or algorithm, and $\bm U \in \mathcal U \subseteq \mathbb R^m$ (for some $m\ge1$) represents model noise or latent variability associated with $\mathcal D$. 
We assume that $\bm U$ can be simulated from a known distribution function $F_{\bm U}(\cdot)$, as in \cite{liang2024} and other BFF work such as \cite{berger2024handbook}. 
The formulation in \eqref{eq:inverted-general} is broad enough to cover nearly all statistical and machine‐learning models used for rank inference. As an illustration, consider a generic generative model in which the random data $\mathcal D = \{(\bm{Y}, \bm{X})\}$ are produced from random noise $\bm U$ given model parameters $\boldsymbol\eta$:
\begin{equation}
\label{eq:model}
\bm{Y} = G(\boldsymbol\eta, \bm{X}, \bm U),
\end{equation}
where $G:\Omega\times\mathcal U\times\mathcal X \to \mathcal Y$ is a mapping and $\boldsymbol\theta = \zeta(\boldsymbol\eta)$. 
Since any sample from a density or mass function $f_{\boldsymbol\eta}(\cdot)$ can be generated via the inverse transform $F^{-1}_{\boldsymbol\eta}(Z)$ with $Z\sim\mathrm{Uniform}(0,1)$, most likelihood‐based models can be represented in the form \eqref{eq:model} (see \cite{xie2022reprosamplesmethodfinite}). 
The sample‐realized version corresponding to \eqref{eq:model} is
\begin{equation}
\label{eq:observed-model}
\bm{y}^{\mathrm{obs}} = G(\boldsymbol\eta^{(0)}, \bm{x}^{\mathrm{obs}}, \bm u^{\mathrm{rel}}),
\end{equation}
where $\bm u^{\mathrm{rel}}$ denotes the realized (unobserved) value of the random noise $\bm U$. 
If $\bm u^{\mathrm{rel}}$ were available, the $K$ target parameters $\boldsymbol\theta^{(0)} = \zeta(\boldsymbol\eta^{(0)})$ could be recovered by solving
\begin{equation}
\label{eq:inverted}
\boldsymbol\theta^{(0)} 
= \arg\min_{\boldsymbol\theta} \Big\{ \min_{\boldsymbol\eta : \zeta(\boldsymbol\eta)=\boldsymbol\theta}
L\big(\bm{y}^{\mathrm{obs}}, G(\boldsymbol\eta, \bm{x}^{\mathrm{obs}}, \bm u^{\mathrm{rel}})\big) \Big\}
\;\overset{\mathrm{def}}{=}\;
H(\mathcal D^{\mathrm{obs}}, \bm u^{\mathrm{rel}}),
\end{equation}
for an appropriate loss $L(\cdot)$, such as the squared‐error loss when $\mathcal Y \subseteq \mathbb R^K$. 
This optimization view shows that, for data generated under \eqref{eq:model}, the assumption \eqref{eq:inverted-general} typically holds. 
When \eqref{eq:observed-model} admits a unique solution in $\boldsymbol\theta$, the optimizer in \eqref{eq:inverted} coincides with $\boldsymbol\theta^{(0)}=\zeta(\boldsymbol\eta^{(0)})$; in more complex cases, \eqref{eq:inverted} may yield a local optimum, in which case our rank inference for $\bm R^{(0)} = \mathcal S(\boldsymbol\theta^{(0)})$ is based on $\boldsymbol\theta^{(0)}$ as defined in \eqref{eq:inverted-general}. 
Finally, the formulation \eqref{eq:inverted-general} also extends beyond generative models, encompassing settings where the observed data cannot be represented in the form \eqref{eq:model}; an example is the nonparametric quantile‐ranking model discussed in Section~3.1.
\subsection{Notations}

Throughout the paper we use the following notations. Lowercase letters (e.g.\ $y$, $u$, $\theta$) denote scalar quantities.
Boldface lowercase letters (e.g.\ $\bm{y}$, $\bm{u}$, $\bm{\theta}$) denote vectors.
Uppercase letters (e.g.\ $Y$, $U$) denote random responses, random variables or vectors.
We write $S_K$ for the set of all permutations of $[K]=\{1,\dots,K\}.$
The symbol $\mathbb{I}\{\cdot\}$ denotes the indicator function, taking value $1$ 
when the condition inside braces is true and $0$ otherwise.
For any random element $A$, we write $\mathbb{P}_{A}(\cdot)$ and $\mathbb{E}_{A}(.)$ for 
probability and expectation under the distribution of $A$ and
$\mathbb{P}_{A\,|\,B}(\cdot)$ and $\mathbb{E}_{A\,|\,B}[\cdot]$ to denote conditional probability and expectation taken with respect to $A$ given $B.$ We denote the random vector corresponding to the $b^{th}$ repro sample by $\bm{U}^{(b)}.$ Joint probability with respect to latent noise $\bm{U}$ and repro samples $\bm{U}^{\ast(1)},\dots,\bm{U}^{\ast(|\mathcal{V}|)}$  
for a fixed index set $\mathcal{V}$, is denoted by $\mathbb{P}_{\bm{U},\mathcal{V}}(\,\cdot\,),$
and the corresponding conditional and unconditional expectations follow the 
same subscript convention.

\section{Methodology Developments and Theories}

From (\ref{eq:rank and theta}) and (\ref{eq:inverted-general}), we can write
%that the set of unknown population ranks can be expressed as  
$
\bm{R} = \mathcal{S}\left( H(\mathcal{D},\bm{U}) \right).$
Thus, if we observe 
$\mathcal{D}^{\text{obs}}$ and knew 
$\bm{u}^{\mathrm{rel}}$,
the true rank in (\ref{eq:rank and theta true}) can be fully recovered as  
\begin{equation}
\label{eq:rank_and_u}
\bm{R}^{(0)} = \mathcal{S}\left( H(\mathcal{D}^{\text{obs}},\bm{u}^{\text{rel}}) \right)
\end{equation}
However, we do not know $ \boldsymbol{u}^{\text{rel}} $, but we know $ F_{\boldsymbol{U}}(\cdot)$ so we can simulate model noise, say $ \boldsymbol{u}^\star $, mimicking $ \boldsymbol{u}^{\text{rel}} $ and use these synthetic $ \boldsymbol{u}^\star $ to help make inference for $ \boldsymbol{R}^{(0)} $.
Following \cite{xie2022reprosamplesmethodfinite}, we refer such ${\bm u}^\star $'s as repro samples of $\bm{U}$. For a repro sample ${\bm u}^\star $ that we generate, we define $ \bm{\theta}^\star =H(\mathcal{D}^{\text{obs}},{\bm u}^\star )$ and 
\begin{equation}
\label{eq:theta_star_def}
{\bm R}^\star = {\cal{S}}(\bm{\theta}^\star) ={\cal{S}}\big(H(\mathcal{D}^{\text{obs}},{\bm u}^\star )\big)
 \end{equation} 
which forms a mapping from ${\bm u}^\star \in {\cal U} \to \bm{R}^\star \in {\cal I}$.   In Section 2.1, we construct a $1-\alpha$ confidence set for the ranks of a selected subset of populations. In Section 2.2 we create a candidate rank set using repro samples ${\bm u}^\star $ that includes the true rank $\bm{R}^{(0)}$ with  high probability that can be used for more complex general models to obtain a  computationally tractable confidence set.
\subsection{Level $1-\alpha$ Confidence Set for Rank Vectors}

 \phantomsection
\noindent\textbf{Inversion argument and neighborhood sets:}\label{subsec:inversion}
Suppose we wish to infer the true rank $r_k^{(0)}$ of population $k$.
If $\theta_i^{(0)}<\theta_k^{(0)}$, then $r_i^{(0)}<r_k^{(0)}$.  
Under the model
$\bm\theta^{(0)} = H(\mathcal D^{\mathrm{obs}},\bm u^{\mathrm{rel}}),$
this ordering implies
$H(\mathcal D^{\mathrm{obs}},\bm u^{\mathrm{rel}})_i
<H(\mathcal D^{\mathrm{obs}},\bm u^{\mathrm{rel}})_k.$
Similarly, if $r_i^{(0)}>r_k^{(0)}$, then
$H(\mathcal D^{\mathrm{obs}},\bm u^{\mathrm{rel}})_i>H(\mathcal D^{\mathrm{obs}},\bm u^{\mathrm{rel}})_k.$
For any noise vector $\bm u\in \mathcal U$, define the neighborhood sets
\begin{align*}
\mathcal N_k^{-}(\mathcal D^{\mathrm{obs}},\bm u)
&=
\bigl\{\, i\neq k:\;
H(\mathcal D^{\mathrm{obs}},\bm u)_i
<
H(\mathcal D^{\mathrm{obs}},\bm u)_k
\bigr\},\\[3pt]
\mathcal N_k^{+}(\mathcal D^{\mathrm{obs}},\bm u)
&=
\bigl\{\, i\neq k:\;
H(\mathcal D^{\mathrm{obs}},\bm u)_i
>
H(\mathcal D^{\mathrm{obs}},\bm u)_k
\bigr\}.
\end{align*}
Thus, for the unknown true noise $\bm u^{\mathrm{rel}}$, 
$|\mathcal N_k^{-}(\mathcal D^{\mathrm{obs}},\bm u^{\mathrm{rel}})|+1
\;\le\;
r_k^{(0)}
\;\le\;
K-|\mathcal N_k^{+}(\mathcal D^{\mathrm{obs}},\bm u^{\mathrm{rel}})|.$ These inequalities,
hold for the unknown true rank $r_k^{(0)}$ conditional on the latent noise 
$\bm u^{\mathrm{rel}}$ that actually generated the data.   If the true noise $\bm u^{\mathrm{rel}}$ were observable, these inequalities would 
identify all rank values consistent with the observed data.   Thus they describe the complete set of rank vectors that are compatible with 
$\mathcal D^{\mathrm{obs}}$ under the (unknown) latent perturbations that produced it. However the true noise $\bm u^{\mathrm{rel}}$ is unobserved, so we replace it by all noise realizations lying in a 
$(1-\alpha)$-probability Borel set 
$B_\alpha(\theta).$ Imposing the inequalities over this region yields a confidence set that contains the true rank vector with probability at least 
$1-\alpha.$

 \phantomsection
\noindent\textbf{General Borel-Set Constraint:}\label{subsec:iGeneral Borel-Set}
Let $T:\mathcal U\times\Theta\to\mathbb R^p$ be a measurable map and let
$B_\alpha(\bm\theta)\subset\mathcal{U}$ satisfy
\begin{equation}
\label{Tuprob}
\mathbb P_{\bm U}\bigl(T(\bm U,\bm\theta)\in B_\alpha(\bm\theta)\bigr)=1-\alpha,
\qquad 0<\alpha<1,
\end{equation}
for every fixed $\bm\theta$.
For each $\bm\theta$, the set $B_\alpha(\bm\theta)$ determines the $(1-\alpha)$-probability region of the latent noise. Next we fix an index set $\mathcal I=\{t_1,\dots,t_{|\mathcal I|}\}\subseteq[K]$.
The restricted subset rank vector of the populations $\mathcal P_{t_1},\ldots,\mathcal P_{t_{|\mathcal I|}}$ is
$\bm R|_{\mathcal I}=(r_{t_l})_{t_l\in\mathcal I}.$
Our goal is to construct a joint confidence set for $\bm R|_{\mathcal I}$.\\
\phantomsection
\noindent\textbf{{Rank Confidence Set}:}\label{subsec:Rank set}
We define the joint confidence set for the ranks of the populations in~$\mathcal I$
\begin{equation}
\label{conf_set_final}
\begin{aligned}
\Gamma_{\alpha}^{\mathcal I}(\mathcal D^{\mathrm{obs}})
=
\Bigl\{\,
\bm R|_{\mathcal I}\;:\;&
\exists\,\bm u^\star\in\mathcal U
\ \text{such that }\
T(\bm u^\star,\bm\theta)\in B_\alpha(\bm\theta),\\[4pt]
&\bigl|\mathcal N^{-}_{t_l}(\mathcal D^{\mathrm{obs}},\bm u^\star)\bigr|+1
\ \le\ r_{t_l}\ \le\
K-\bigl|\mathcal N^{+}_{t_l}(\mathcal D^{\mathrm{obs}},\bm u^\star)\bigr|,\ 
\forall\, t_l\in\mathcal I
\Bigr\}.
\end{aligned}
\end{equation}
The existential quantifier reflects the  principle: a rank vector $\bm R|_{\mathcal I}$ is included in 
$\Gamma_{\alpha}^{\mathcal I}(\mathcal D^{\mathrm{obs}})$ 
whenever there exists at least one noise realization 
$T(\bm u^\star,\bm\theta)\in B_\alpha(\bm\theta)$ for which the ordering 
constraints implied by the observed data are satisfied.

\phantomsection
\noindent\textbf{Finite-Sample Validity:} \label{subsec:Finite-Sample Validity} The following result shows that the rank set construction in \eqref{conf_set_final} attains at least $(1-\alpha)$ coverage for the true restricted rank vector.

\begin{theorem}
\label{thm_rank_confidence_jrssb}
Let $\bm R|_{\mathcal I}^{(0)}$ be the true rank vector for the populations indexed by $\mathcal I$.
If the model $\bm R^{(0)}=\mathcal S\bigl(H(\mathcal D,\bm U)\bigr)$ holds and the Borel-set constraint \eqref{Tuprob} is exact, then
$$\mathbb P_{\bm U}\bigl(
\bm R|_{\mathcal I}^{(0)}\in \Gamma_{\alpha}^{\mathcal I}(\mathcal D)
\bigr)
\;\ge\; 1-\alpha.$$
More generally if, $\mathbb P_{\bm U}\bigl(T(\bm U,\bm\theta)\in B_\alpha(\bm\theta)\bigr)
\ge (1-\alpha)\bigl(1+o(\delta')\bigr),$
then for  $\delta'>0$ which may depend on $\sum_{k=1}^K n_k,$ we have $\mathbb P_{\bm U}\bigl(
\bm R|_{\mathcal I}^{(0)}\in \Gamma_{\alpha}^{\mathcal I}(\mathcal D)
\bigr)
\ge (1-\alpha)\bigl(1+o(\delta')
\bigr),$

\end{theorem}

\phantomsection
\noindent\textbf{Interpretation:}\label{subsec:Interpretation}
Although $\bm{u}^{\mathrm{rel}}$ is unknown, we have size $1-\alpha$ confidence that $T\bigl(\bm{u}^{\mathrm{rel}},\bm{\theta}\bigr)\in B_{\alpha}(\bm{\theta}).$ As we set 
$\bm{R}^{(0)}= \mathcal{S}\!\bigl(H(\mathcal{D}^{\text{obs}},\,\bm{u}^{\mathrm{rel}})\bigr)$ it follows that $\bm R|_{\mathcal I}^{(0)} \in\Gamma^{\mathcal{I}}_{\alpha}(\mathcal{D}^{\text{obs}})$.
More generally, under any model of the form 
$\bm{R}^{(0)}=\mathcal{S}\bigl(H(\mathcal{D},\,\bm{U})\bigr)$, the event
$\bigl\{T(\bm{U},\bm{\theta})\in B_{\alpha}(\bm{\theta})\bigr\}\subseteq
  \bigl\{\bm R|_{\mathcal I}^{(0)} \in\Gamma^{\mathcal{I}}_{\alpha}(\mathcal{D})\bigr\}.$
Hence,
$1 - \alpha \le
  \mathbb{P}_{\bm{U}}\bigl(T(\bm{U},\bm{\theta})\in B_{\alpha}(\bm{\theta})\bigr)
  \;\le\;
\mathbb{P}_{\bm{U}}\bigl(\bm{R}|^{(0)}_{\mathcal I}\in \Gamma^{\mathcal{I}}_{\alpha}(\mathcal{D})\bigr)$
demonstrating that $\Gamma_{\alpha}^{\mathcal{I}}(\mathcal{D}^{\text{obs}})$ indeed achieves at least $1-\alpha$ coverage probability for any ${\mathcal{I}}.$ Because $\bm u^{\mathrm{rel}}$ and the unknown $\bm u^{\star}$ are identically distributed, if $T\bigl({\bm u}^\star ,\bm{\theta}\bigr)$  confined within the same set  $ B_\alpha(\bm{\theta})$, we get $1-\alpha$ coverage. \\
\phantomsection
\noindent\textbf{Example 2.1: Independent Gaussian Populations:}\label{subsec:Example 2.1}
Consider $K$ independent Gaussian populations with
$y_{ik}^{\mathrm{obs}}$ from $N(\theta^{(0)}_k,\sigma_k^2),$ for a known $\sigma_k,k\in[K].$
Let $n_k$ observations be drawn from each population and define the sample means
$y_k^{\mathrm{obs}}
=\frac1{n_k}\sum_{i=1}^{n_k}y_{ik}^{\mathrm{obs}}.$
Let $u_k^{\mathrm{rel}}$ be a realization from $N(0,1)$.  
By the Gaussian location-scale identity,
$\theta_k^{(0)}
=
y_k^{\mathrm{obs}}
-
\frac{\sigma_k^{(0)}}{\sqrt{n_k}}\,
u_k^{\mathrm{rel}},$ $k=1,\dots,K.$
We choose  $$T(\bm{u} ,\bm{\theta})=\left(\min_{i \in [K]} \left( \frac{\sigma_i u_i}{\sqrt{n_i}} - \frac{\sigma_k u_k}{\sqrt{n_k}} \right),\max_{i \in [K]} \left( \frac{\sigma_i u_i}{\sqrt{n_i}} - \frac{\sigma_k u_k}{\sqrt{n_k}} \right)\right)$$ and $ B_\alpha(\boldsymbol{\theta}) $ such that
$ \mathbb{P}_{\bm{U}}\left\{ 
\max_{i \in [K]} \left( \frac{\sigma_i u_i}{\sqrt{n_i}} - \frac{\sigma_k u_k}{\sqrt{n_k}} \right) < c^{+}_k, 
\min_{i \in [K]} \left( \frac{\sigma_i u_i}{\sqrt{n_i}} - \frac{\sigma_k u_k}{\sqrt{n_k}} \right) > c^{-}_k ,\forall k
\right\} = 1 - \alpha$
for suitable $c_k^{+}>c_k^{-}>0$.
On the event in  $B_\alpha(\boldsymbol{\theta})$ in \eqref{Tuprob} which holds with probability $1-\alpha$ we have 
\[
c_k^{+} \;<\;  \frac{\sigma_i u_i}{\sqrt{n_i}} - \frac{\sigma_k u_k}{\sqrt{n_k}}
\;<\; c_k^{-}\qquad\text{for all } i\neq k .
\]
Thus whenever the observed difference \( y_i^{\mathrm{obs}} - y_k^{\mathrm{obs}} \) is less than the lower tolerance bound \( c_k^{-} \), it follows that $\theta_i^{(0)} - \theta_k^{(0)}< 0.$ Similarly, whenever the observed difference \( y_i^{\mathrm{obs}} - y_k^{\mathrm{obs}} \) exceeds the negative upper tolerance bound \( -c_k^{+} \), we have $\theta_i^{(0)} - \theta_k^{(0)}> 0.$ Define the neighborhood sets as $\mathcal{N}^{-}_k(\mathcal{D}^{\text{obs}}, \bm{u}) =
\left\{ i \ne k:\;
  {y}_i^{\mathrm{obs}} -   {y}_k^{\mathrm{obs}} < c^-_k \right\},$ $
\mathcal{N}^{+}_k(\mathcal{D}^{\text{obs}},\bm{u}) =
\left\{ i \ne k:\;
  {y}_i^{\mathrm{obs}} -   {y}_k^{\mathrm{obs}} > -c_k^+ \right\}.$ Then on $B_\alpha(\theta),$ $i\in\mathcal{N}^{-}_k(\mathcal{D}^{\text{obs}}, \bm{u}^{\mathrm{rel}}) $ implies $\theta_i^{(0)} < \theta_k^{(0)},$ $
i\in\mathcal{N}^{+}_k(\mathcal{D}^{\text{obs}}, \bm{u}^{\mathrm{rel}}) $ implies $\theta_i^{(0)} > \theta_k^{(0)}.$ Then a corresponding level $1-\alpha $ confidence set for the rank vector $ \bm{R} $  is given by $
\Gamma^{\mathcal{I}}_\alpha\bigl(\mathcal{D}^{\mathrm{obs}}\bigr)$ defined in (\ref{conf_set_final}).

\phantomsection
\noindent\textbf{Monte Carlo approximation: }\label{subsec:Monte Carlo}When the distribution of $T(\bm U,\bm\theta)$ is not available analytically, we approximate $B_\alpha(\bm\theta)$ as follows. This yields a Monte-Carlo approximation to $\Gamma_\alpha^{\mathcal I}(\mathcal D^{\mathrm{obs}})$ while preserving the finite-sample guarantee of Theorem~\ref{thm_rank_confidence_jrssb}.

\begin{algorithm}[H]
\caption{Monte-Carlo Construction of the Repro-Samples Rank Confidence Set $\Gamma_{\alpha}^{\mathcal I}(\mathcal D^{\mathrm{obs}})$}
\label{alg:construction}
\begingroup
\renewcommand{\baselinestretch}{.9}\normalsize
\setlength{\belowdisplayskip}{0pt}
\setlength{\abovedisplayskip}{0pt}
\setlength{\intextsep}{0pt}
\setlength{\textfloatsep}{0pt}
\setlength{\floatsep}{0pt}
\setlength{\parskip}{0pt}
\setlength{\itemsep}{0pt}
\begin{algorithmic}
\State \textbf{Step 1.} For a given parameter value $\bm\theta\in\Theta$, compute $\bm{R}=(r_1,..r_K)=\mathcal{S}(\bm{\theta})$: 
\Statex \hspace{1em}\textit{(a)} Generate $\bm U^S\in \mathcal{U},$ $s=1,..,B$ and use the Monte Carlo method based on the finite set $\{T(\bm U^s,\bm\theta),s=1,..,B\}$ to obtain the level-$\alpha$ Borel set $B_\alpha(\bm\theta)$ in~\eqref{Tuprob} from the empirical distribution of $T(\bm U,\bm\theta).$
\Statex \hspace{1em}\textit{(b)} If there exists $\bm u^{\star}\in \mathcal{U}$, check whether $T(\bm u^{\star},\bm\theta)\in B_\alpha(\bm\theta)$ and  $|\mathcal N_k^{-}(\mathcal D^{\mathrm{obs}},\bm u^{\star})|+1<r_k<K-|\mathcal N_k^{+}(\mathcal D^{\mathrm{obs}},\bm u^{\star})|$ for $k\in[K]$. If both of the above criteria are satisfied, keep the $\bm{R}$.
\State \textbf{Step 2.}  Collect all kept $\bm{R}$ and subset $\bm{R}|_\mathcal{I}$ to form the confidence set $\Gamma_{\alpha}^{\mathcal I}(\mathcal D^{\mathrm{obs}}).$
\end{algorithmic}
\endgroup
\end{algorithm}

\subsection{Refined confidence set using a candidate set with high coverage}
\label{subsec:refined_confidence_set}

This section develops a reduction of the rank search space through a data–adapted
candidate set that retains all rank vectors compatible with the observed data and with
high–probability neighborhoods of the latent noise vector.

\phantomsection
\noindent\textbf{Construction of the candidate set:}
\label{subsec:CandidateSet} The mapping
$\bm u^\star \mapsto 
\bm R^\star = \mathcal S\!\big(H(\mathcal D^{\mathrm{obs}},\bm u^\star)\big)$
is typically \textit{many–to–one}, the latent noise space $\mathcal U$ is (often) uncountable,
whereas the rank space $S_K$ contains $K!$ permutations.
Hence many distinct noise vectors may induce the same rank vector.  
In principle, the ranks produced by a large number of repro samples could range
over all of $S_K$, but for moderate or large $K$ this becomes computationally infeasible.
Thus it is desirable to restrict attention to a subset of $S_K$
that still contains the truth with high probability.
To quantify compatibility with the empirical ordering, for each ordered pair $(i,j)$
we test whether the repro sample reverses the empirical order.  
The discordance statistic aggregates such reversals:
\[
\mathrm{Disc}(\mathcal D^{\mathrm{obs}},\bm\theta^\star)
=
\sum_{i\neq j}
\mathbb I\!\Big(
(\hat\theta^{\mathrm{obs}}_i-\hat\theta^{\mathrm{obs}}_j)
(\theta^\star_i-\theta^\star_j)
<0
\Big).
\]
Fix a discordance budget $c>0$ and define the candidate set
\begin{equation}
\label{eq:candidateset}
\mathcal C_{\mathcal{V}}(\mathcal{D}^{\mathrm{obs}})
=
\Bigl\{
\bm R^\star :
\bm R^\star=\mathcal S(\bm\theta^\star),\
\mathrm{Disc}(\mathcal D^{\mathrm{obs}},\bm\theta^\star)<c,\
\bm u^\star\in\mathcal V
\Bigr\},
\end{equation}
where 
$\mathcal V=\{\bm u^{\star(1)},\ldots,\bm u^{\star(|\mathcal V|)}\}$ 
are i.i.d.\ draws from $F_{\bm U}$, independent of $\bm u^{\mathrm{rel}}$.
Smaller $c$ yields a tighter candidate set, while larger $c$ yields increased robustness. \\
\phantomsection
\noindent\textbf{Example 2.1 continued.} 
We generate \( |\mathcal V| \) i.i.d.\ perturbations \(u_k^{\ast(b)}\)from \(N(0,1)\) and form
$\theta_k^{\ast(b)}
 \;=\; y_k^{\mathrm{obs}}
   - \frac{\sigma_k^{(0)}}{\sqrt{n_k}}\,u_k^{\ast(b)},$ for each $k.$
Since $y_k^{\mathrm{obs}}=\frac{1}{n_k}\sum_{i=1}^{n_k} y_{ik}^{\mathrm{obs}}$
is an unbiased and consistent estimator of $\theta_k^{(0)},$ we use $y_k^{\mathrm{obs}}$ to define the discordance as
$\mathrm{Disc}\left(\mathcal D_n^{\mathrm{obs}},\bm\theta^{\ast(b)}\right)=
 \sum_{i\neq j}\mathbb{I}(
 \bigl(y_i^{\mathrm{obs}} - y_j^{\mathrm{obs}}\bigr)
 \bigl(\theta_i^{\ast(b)} - \theta_j^{\ast(b)}\bigr)<0),$
and then the candidate set \(\mathcal C_{\mathcal{V}}(\mathcal{D}_n^{\mathrm{obs}})\) is obtained as in~\eqref{eq:candidateset}.\\
A key observation is that the true latent noise $\bm u^{\mathrm{rel}}$ generates the
true ranking: $\bm R^{(0)}
=\mathcal S\!\big(H(\mathcal D^{\mathrm{obs}},\bm u^{\mathrm{rel}})\big).$ Consequently, any repro sample $\bm u^\star$ that lies sufficiently close to 
$\bm u^{\mathrm{rel}}$ must produce the same ordering.
Intuitively, if we generate many independent draws of $\bm u^\star$, at least one
should fall in a neighbourhood of $\bm u^{\mathrm{rel}}$ and therefore replicate
$\bm R^{(0)}$.

\phantomsection
\noindent\textbf{Neighborhood condition:}
\label{subsec:NeighborhoodCondition} Formally, suppose there exists a neighborhood 
$Q_n(\bm u^{\mathrm{rel}})\subseteq\mathcal U$ such that every $\bm u^\star\in Q_n(\bm u^{\mathrm{rel}})$ satisfies
$\mathcal S\!\big(H(\mathcal D^{\mathrm{obs}},\bm u^\star)\big)=\bm R^{(0)}.$ We assume an independent latent-noise copy $\bm u^\star$ falls in this neighborhood with positive probability:
\begin{equation}
\label{A1}
\tag{A1}
\mathbb P_{\bm U,\bm U^\star}\!\big(\bm U^\star\in Q_n(\bm U)\big)\ \ge\ c_n>0.
\end{equation}
Condition~(A1) formalizes the idea that the true ordering reappears among the repro
samples with non-negligible probability.

\phantomsection
\noindent\textbf{Guaranteeing coverage:}
\label{subsec:guaranteeingcoverage} We now establish that the true rank vector $\bm R^{(0)}$ lies in the
candidate set $\mathcal C_{\mathcal{V}}(\mathcal{D}^{\mathrm{obs}})$ with high probability.
The first step is to show that the true parameter
$\bm\theta^{(0)}$ falls in the feasible region
$\{\bm\theta:\mathrm{Disc}(\mathcal D,\bm\theta)<c\}$ with
non-negligible probability.  The argument is based on pairwise reversal
probabilities and a simple Markov bound, the proof is given in the
Appendix.
\begin{lemma}\label{thm:disc-markov-ordered}
Let $\bm{\theta}^{(0)}\in\mathbb{R}^K$ be fixed, consider the model $\bm{\theta}^{(0)} 
= H(\mathcal{D},\bm{U})$ and let $\hat{\bm{\theta}}\in\mathbb{R}^K$ be any reasonable estimator.
For $1\le i\ne j\le K$, define the gap $\Delta^{(0)}_{ij}=\theta^{(0)}_i-\theta^{(0)}_j$, the gap between estimator $\widehat\Delta_{ij}=\hat\theta_i-\hat\theta_j$, $\delta_{ij}=\widehat\Delta_{ij}-\Delta^{(0)}_{ij}$ and $p_{ij} = \mathbb{P}_{\bm{U}}\big(|\delta_{ij}|\ \ge\ |\Delta^{(0)}_{ij}|\big).$ Then for any $c>0,$ \vspace{-3mm}
\begin{equation}\label{eq:disc-basic-ordered}
\mathbb{P}_{\bm{U}}\big\{\bm{\theta}^{(0)}\notin \{\bm{\theta}:Disc(  \mathcal{D},\bm{\theta})<c\}\}=\mathbb{P}_{\bm{U}}\!\big\{Disc(  \mathcal{D},\bm{\theta}^{(0)})\ge c\big\}
\ \le\
\frac{1}{c}\sum_{1\le i\ne j\le K} p_{ij}
\end{equation} 
 \noindent Moreover: {\bf (a) Finite Variance.}
If $Variance(\delta_{ij})\le m_{ij}^2$, then $p_{ij} \le \frac{m_{ij}^2}{\Delta^{(0)^2}_{ij}}$,\\ $ \mathbb{P}_{\bm{U}}\big\{\bm{\theta}^{(0)}\notin \{\bm{\theta}:Disc(  \mathcal{D};\bm{\theta})<c\}\} \le
\frac{1}{c}\sum_{i\ne j}\frac{m_{ij}^2}{\Delta^{(0)^2}_{ij}}$\\
\noindent{\bf (b) Sub-Gaussian}
If $\delta_{ij}$ is sub-Gaussian with parameter $\tau_{ij}^2$, that is  
$\mathbb{E}_{\bm{U}}\big[e^{\lambda \delta_{ij}}\big] \le \exp\!\Big(\frac{\lambda^2\tau_{ij}^2}{2}\Big)$ for all $\lambda\in\mathbb{R},$ % 
then $\mathbb{P}_{\bm U}\!\Big\{
\bm{\theta}^{(0)}\notin 
\{\bm{\theta}:\mathrm{Disc}(\mathcal D,\bm\theta)<c\}
\Big\}\le\frac{1}{c}\sum_{i\ne j}2\exp\left(-\frac{(\Delta_{ij}^{(0)})^2}{2\tau_{ij}^2}\right)$ 
\\ In addition, if $\Delta^{(0)}_{\min}=\min_{i\ne j}|\Delta^{(0)}_{ij}|>0$  $\tau_{ij}\le\tau$, $i\ne j$, the shortfall from $1$ decays exponentially in the pairwise signal-to-noise ratio and 
$\mathbb{P}_{\bm{U}}\big\{\bm{\theta}^{(0)}\notin \{\bm{\theta}:Disc(  \mathcal{D};\bm{\theta})<c\}\}\le\
\frac{2K(K-1)}{c}\exp\!\Big(-\frac{\Delta_{\min}^{{(0)^2}}}{2\tau^2}\Big)$ 
\end{lemma}

\paragraph{High–probability inclusion in the candidate set:} Using Lemma~\ref{thm:disc-markov-ordered}, we next show that if
$|\mathcal V|$ is sufficiently large, then at least one repro draw
aligns with $\bm R^{(0)}$, ensuring the true rank lies in the candidate
set defined in~\eqref{eq:candidateset}.  Proof details are given in the Appendix.
\begin{lemma}\label{coverage-disc}Let $\mathcal V=\{\bm u^{\star(1)},\ldots,\bm u^{\star(|\mathcal V|)}\}$ denote $|\mathcal V|$ draws from $F_{\bm U}(\cdot)$, and suppose Assumption~\textup{(A1)} holds. 
Define 
$q_n=\mathbb{P}_{\bm U}\!\big\{\mathrm{Disc}(\mathcal D,\bm\theta^{(0)})<c\big\}$. Then for a positive constant $c_0>0,$ the candidate set $\mathcal C_{\mathcal{V}}(\mathcal{D}^{\mathrm{obs}})$ defined in~\eqref{eq:candidateset} is such that
$\mathbb{P}_{\bm U,\mathcal V}\!\left(\bm R^{(0)}\notin \mathcal C_{\mathcal{V}}(\mathcal{D})\right)
\ \le\ 1-q_n+e^{-c_0|\mathcal V|}$
\end{lemma}
\phantomsection
\noindent\textbf{Choice of c:} 
We choose the threshold $c$ so that the probability 
$q_n=\mathbb P_{\bm U}\!\bigl(
\mathrm{Disc}(\mathcal D^{\mathrm{obs}},\bm\theta^{(0)})<c\bigr)$ is close to one.  
Operationally, we fix a target $q_n\in\{0.90,0.95\}$, generate i.i.d.\ draws 
$\bm u^{\star(b)}$ from $ F_{\bm U}$, construct $\bm\theta^{\star(b)}$, compute 
$\mathrm{Disc}(\mathcal D^{\mathrm{obs}},\bm\theta^{\star(b)})$ for 
$b=1,\dots,B$, and set $c$ to the empirical 90th--95th percentile of these 
discordance values.  
This retains nearly all oracle-like $\bm \theta^\star.$
In block-independent models where observations for differents populations are independent, $q_n$ can be made explicit. For example, Lemma~\ref{thm:disc-markov-ordered}(b) shows that
$\mathbb P_{\bm U}\!\left(
(\widehat\theta_i-\widehat\theta_j)
(\theta_i^{(0)}-\theta_j^{(0)})<0
\right)
\le 
\exp\!\left\{-\frac{(\Delta_{ij}^{(0)})^2}{2\tau_{ij,n}^2}\right\},$
so the expected discordance fraction satisfies
$\frac{\mathbb E[\mathrm{Disc}]}{K_{\mathrm{pairs}}}\le
\exp\!\left\{
-\frac{(\Delta_{\min}^{(0)})^2}{2\tau^2}
\right\},$ where $K_{\mathrm{pairs}}=\binom{K}{2}.$
This quantity is exponentially small in the minimum signal-to-noise ratio.  
We estimate it empirically using
$\widehat{\mathrm{SNR}}_{\min}
=
\min_{i\neq j}
\frac{|\widehat\Delta_{ij}|}{\hat\tau_{ij}},$
$\widehat{\bar p}_{\mathrm{disc}}
=\exp\!\left(-\frac{\widehat{\mathrm{SNR}}_{\min}^2}{2}\right).$ McDiarmid’s inequality yields $q_n
=
\mathbb P_{\bm U}\!\left(
\mathrm{Disc}(\mathcal D^{\mathrm{obs}},\bm\theta^{(0)})
<
p^\ast K_{\mathrm{pairs}}
\right)
\ge 
1-\exp\{-c^\star K\varepsilon^2\},$
$\varepsilon = p^\ast - \widehat{\bar p}_{\mathrm{disc}},$
so $q_n\approx 1$ whenever $p^\ast > \widehat{\bar p}_{\mathrm{disc}}$.  
Motivated by this bound, we choose 
$p^\ast=\lambda\,\widehat{\bar p}_{\mathrm{disc}},$
$\lambda\approx 1.2\text{--}1.5.$
Finally, we set the discordance cutoff to
$c=\lfloor p^\ast K_{\mathrm{pairs}}\rfloor,$ a simple and data-adaptive threshold ensuring that the repro-samples retained in 
$\mathcal C_{\mathcal V}(\mathcal D^{\mathrm{obs}})$ remain consistent with the observed ordering with high probability.

\phantomsection
\noindent{\textbf{Refined confidence set:}}
\label{subsec:refinedset}We refine the the rank confidence set
$\Gamma_\alpha^{\mathcal I}(\mathcal D^{\mathrm{obs}})$
by intersecting it with the candidate set as
$\tilde{\Gamma}_{\mathcal{V}_\alpha}^{\mathcal I}(\mathcal D^{\mathrm{obs}})
=
\Gamma_{\alpha}^{\mathcal I}(\mathcal D^{\mathrm{obs}})
\cap
\mathcal C_{\mathcal{V}}(\mathcal{D}^{\mathrm{obs}}).$
This removes rank vectors incompatible with any low discordance repro sample.

\begin{corollary}
Let $\bm R|_{\mathcal{I}}^{(0)}$ denote the true rank vector for the populations 
$\{\mathcal P_{t_\ell}: t_\ell\in\mathcal I\}$. Assume that the model \eqref{eq:rank_and_u} holds and that the Borel--set condition 
\eqref{Tuprob} is exact for every $\bm\theta$. Let 
$q_n=\mathbb P_{\bm U}\!\left\{\mathrm{Disc}(\mathcal D,\bm{\theta}^{(0)}) < c \right\},$
and suppose $1-q_n \le \zeta.$ 
Then, for some constant $c_0>0,$ the set
$\tilde{\Gamma}_{\mathcal{V}_\alpha}^{\mathcal I}(\mathcal D)
=
\Gamma^{\mathcal I}_{\alpha}(\mathcal D)
\;\cap\;
\bigl\{\bm R^\star : \bm R^\star \in C_{\mathcal{V}}(\mathcal{D})\bigr\},$
where $\Gamma^{\mathcal I}_{\alpha}(\mathcal D)$ is defined in 
\eqref{conf_set_final} satisfies, 
$$\mathbb P_{\bm U,\mathcal V}\!\left(
\bm R|_{\mathcal{I}}^{(0)} \in \tilde{\Gamma}_{\mathcal{V}_\alpha}^{\mathcal I}(\mathcal D)
\right)
\;\ge\;
1 - \alpha - \zeta - e^{-c_0 |\mathcal V|}$$
Moreover, if 
$\mathbb P_{\bm U}\!\{\,T(\bm U,\bm\theta)\in B_\alpha(\bm\theta)\,\}
\;\ge\;
(1-\alpha)\{1+o(\delta')\},$
for some $\delta'>0$ which may or not may not depend on $\sum_{k=1}^K n_k$ we get 
$\mathbb P_{\bm U,\mathcal V}\!\left(
 \bm R|_{\mathcal{I}}^{(0)} \in \tilde{\Gamma}_{\mathcal{V}_\alpha}^{\mathcal I}(\mathcal D)
\right)
\;\ge\;
(1-\alpha)\{1+o(\delta')\}
-\zeta
-e^{-c_0 |\mathcal V|}.$
\end{corollary} 
A proof is given in the Appendix. As long as the discordance filter excludes the true parameter with probability
at most $\zeta$, the set $\tilde{\Gamma}_{\mathcal{V}_\alpha}^{\mathcal I}(\mathcal D^{\mathrm{obs}})$ preserves near-nominal
coverage, $1-\alpha-\zeta - o(1).$ Moreover, by choosing the Borel region to have level $1-\alpha'$ with
$\alpha'$ arbitrarily close to $\alpha$, selecting the discordance tolerance so
that $\zeta$ is negligible, and taking $|\mathcal V|$ sufficiently large so
that the Monte--Carlo error term $e^{-c_0|\mathcal V|}$ vanishes, the three
sources of error can be made arbitrarily small.  Consequently, the coverage of
$\tilde{\Gamma}_{\mathcal{V}_\alpha}^{\mathcal I}(\mathcal D^{\mathrm{obs}})$ can be tuned to approach $1-\alpha$.
The method for obtaining the final confidence set is summarized in Algorithm \ref{alg:candidate}.
\begin{algorithm}[H]
\caption{Candidate‐Adjusted Rank Confidence Set $\tilde{\Gamma}_{\mathcal{V}_\alpha}^{\mathcal I}(\mathcal D^{\mathrm{obs}})$}
\label{alg:candidate}
\begingroup
\renewcommand{\baselinestretch}{0.92}\normalsize
\setlength{\belowdisplayskip}{0pt}
\setlength{\abovedisplayskip}{0pt}
\setlength{\intextsep}{2pt}
\setlength{\textfloatsep}{2pt}
\setlength{\floatsep}{2pt}
\setlength{\parskip}{0pt}
\setlength{\itemsep}{1pt}
\begin{algorithmic}[1]

\State \textbf{Step 1.}  Apply Algorithm~\ref{alg:construction} to the observed data $\mathcal D^{\mathrm{obs}}$ to obtain  
$\Gamma_{\alpha}^{\mathcal I}(\mathcal D^{\mathrm{obs}}),$
the level $1-\alpha$ Repro-Samples rank confidence set derived from $\bm{u}^{\star}\in{\mathcal U}$.
\State \textbf{Step 2.}  
Intersect the base set with the candidate set $C_{\mathcal{V}}(\mathcal{D}^{\mathrm{obs}})$ obtained from a new independent set of  $\{\bm u^{*(i)}\}_{i=1}^{\mathcal V}$ from the same $F_{\bm{U}}(\cdot)$ and obtain $\tilde{\Gamma}_{\mathcal{V}_\alpha}^{\mathcal I}(\mathcal D^{\mathrm{obs}})=
\Gamma_{\alpha}^{\mathcal I}(\mathcal D^{\mathrm{obs}})\cap C_{\mathcal{V}}(\mathcal{D}^{\mathrm{obs}})$
\end{algorithmic}
\endgroup
\end{algorithm}

\subsection{Expected size of the Sub-Gaussian candidate set and discussion}

The coverage guarantees in  
Lemma~\ref{thm:disc-markov-ordered} and Lemma~\ref{coverage-disc} ensure that the true rank vector 
$\bm R^{(0)}$ is contained in $\mathcal C_{\mathcal{V}}(\mathcal{D}^{\mathrm{obs}})$ 
with high probability.  
To assess the computational cost of the refined confidence set 
$\tilde{\Gamma}_{\mathcal{V}_\alpha}^{\mathcal I}(\mathcal D^{\mathrm{obs}})$, it is therefore 
important to understand the typical size of the candidate set.  
This subsection derives an explicit upper bound on 
$\mathbb E_{\bm U,\mathcal V}\!\big[\,|\mathcal C_{\mathcal{V}}(\mathcal{D})|\,\big]$ 
under sub-Gaussian assumptions for both the plug-in estimation error and the 
latent noise, analogous calculations can be carried out for other light-tailed models with minor modifications.

\begin{lemma} 
\label{expectsize}
Let $\{\bm U^{*(b)}\}_{b=1}^{|\mathcal V|}$ be i.i.d.\ draws from a distribution $F_{\bm U}(.)$, independent of $\mathcal D$, and set
$\bm\theta^{*(b)} \;=\; H\!\big(\mathcal D,\,\bm U^{*(b)}\big), $
$\bm R^{*(b)} \;=\; S\!\big(\bm\theta^{*(b)}\big).$
For $i\neq j$ define
$\Delta_{ij}^{(0)}=\theta_i^{(0)}-\theta_j^{(0)},$ $\Delta_{\min}^{(0)}=\min_{i\neq j}|\Delta_{ij}^{(0)}|,$ 
$\hat{\varepsilon}_{ij}=(\hat\theta_i-\hat\theta_j)-\Delta_{ij}^{(0)},$ $
\delta_{ij}^{*(b)}=(\theta_i^{*(b)}-\theta_j^{*(b)})-\Delta_{ij}^{(0)}.$ For any ranking $\bm{R}=(r_1,..,r_K)\in S_K$, define the ordered-pair normalized discordance 
\[
g(\bm{R})\;=\;\frac{1}{2K_{\mathrm{pairs}}}\sum_{i< j}\mathbf 1\!\Big\{\big({\hat{\theta}_i}-{\hat{\theta}_j}\big)(r_i-r_j)\, < 0\Big\}.
\]If the following assumptions hold
\begin{enumerate}
\item[(B1)] (Sub-Gaussian estimate) For each pair $(i,j)$, $\widehat\varepsilon_{ij}$ is mean-zero sub-Gaussian with proxy $v_{ij,n}^2$; that is 
for all $ i\neq j,$  
$E_{\bm{U}}\!\left[\exp\{t\,\widehat\varepsilon_{ij}\}\right]
\;\le\;
\exp\!\left(\tfrac{t^2 v_{ij,n}^2}{2}\right)$ for $t\in\mathbb{R}.$ Define $\bar v_n^2=\max_{i\ne j} v_{ij,n}^2.$
\item[(B2)] (Sub-Gaussian repro errors) Conditionally on $\bm{U},$ each $\delta_{ij}^{*(b)}$ is mean-zero sub-Gaussian with proxy $\sigma_{ij}^2%\le \bar v_n^2
$ and 
$\mathbb E_{\bm{U}^{\star(b)}|\bm{U}}[\exp\{t\delta_{ij}^{*(b)}\}]\le \exp\{t^2\sigma_{ij}^2/2\}$ for all $t\in \mathbb{R}.$ Define $\bar{\tau}_n^2=\max_{i\ne j}\sigma^2_{ij}.$.
\item[(B3)] (Disjoint-pair independence) If $(i,j)$ and $(k,\ell)$ are have no indices in common then $\delta_{ij}^{*(b)}$ and $\delta_{k\ell}^{*(b)}$ are independent given $\bm{U}.$
\end{enumerate} 
Then the expected size of the candidate set satisfies
\begin{equation}
\mathbb{E}_{\bm{U},\mathcal{V}}\!\big[\,|\mathcal C_{\mathcal{V}}(\mathcal{D}^{\mathrm{obs}})|\,\big]
\;\le\;
\Big|\{\,\bm{R}:\,g(\bm{R})\le \tilde g_n\,\}\Big|
\;+\;
\sum_{\bm{R}:\,g(\bm{R})>\tilde g_n}
\exp\!\Big\{-\frac{\Delta_{\min}^{(0)\,2}K_{\mathrm{pairs}}}{w_0 \bar\tau_n^2}\,\big(g(\bm{R})-\tilde g_n\big)\Big\}
\label{eq:expected_size_bound_subg_sharp}
\end{equation}
where $w_0\in\{1,\dots,K\}$ is the edge-coloring constant of the ordered-pair 
graph,  $c$ is the discordance budget and 
$\tilde g_n
=
\frac{
\displaystyle 
\frac{c}{2}
\;+\;
\log\!\Bigl[
\bigl(\tfrac{c}{2}+1\bigr)
\bigl(\tfrac{c}{2}\bigr)^{\,c/2}
\Bigl(1+c+c^2 e^{-\Delta_{\min}^{(0)2}/(8\bar v_n^2)}\Bigr)
\Bigr]
\;+\;
\log|\mathcal V|
}{
\displaystyle 
\frac{\Delta_{\min}^{(0)2}}{w_0 \bar\tau_n^2}K_{\mathrm{pairs}}
}$ a positive constant.
\end{lemma}

\noindent\textbf{Interpretation:}
Lemma~\ref{expectsize} partitions the permutation space into two regimes. A \emph{plausible region}
$\mathcal R_{\mathrm{plaus}}=\{\bm R: g(\bm R)\le \tilde g_n\},$ which contributes deterministically to the expected size. An \emph{implausible region}
$\mathcal R_{\mathrm{impl}}=\{\bm R: g(\bm R)>\tilde g_n\},$ whose contribution is exponentially suppressed at rate $\bigl(\Delta_{\min}^{(0)2 }K_{\mathrm{pairs}}/(w_0\bar\tau_n^2)\bigr)$. Consequently, although $S_K$ contains $K!$ permutations, the effective support 
of the repro-sampling distribution is concentrated around rankings near the 
empirical ordering.  
Rankings with discordance exceeding $\tilde g_n$ have negligible probability 
of appearing in $\mathcal C_{\mathcal{V}}(\mathcal{D}^{\mathrm{obs}})$.

\medskip
\noindent\textbf{Behaviour of the cutoff $\tilde g_n$: }
The threshold $\tilde g_n$ increases with the noise levels $\bar v_n$ and 
$\bar\tau_n$, and with the number of repro-samples $|\mathcal V|$.  
It decreases with the minimal signal $\Delta_{\min}^{(0)}$ and with the number 
of pairwise comparisons $K_{\mathrm{pairs}}$.  
Thus, for moderate $K$ and non-vanishing gaps between $\theta^{(0)}_i$, $\tilde g_n$ remains small, so 
the expected candidate-set size is typically far below $K!$ and the refined 
confidence set remains computationally feasible.

\noindent\textbf{Choice of T(.):} In the repro‑samples framework, one chooses $T(\bm{U},\bm{\theta})$ so that under the true noise $\bm{U}\sim F_{\bm{U}}(\cdot)$, the random vector $T(\bm{U},\bm{\theta})$ has a known distribution (independent of the observed data).  For example, in a quantile  model (Section 3.1) one may simply take $ T(\bm{U},\bm{\theta})=\bm{U}\,$ since $\bm{U}$ itself is binomial and fully characterizes the randomness in the model. More generally, $T(\bm{U},\bm{\theta})$ can be any pivot or likelihood‑ratio–type statistic whose distribution $F_{\bm{U}}(\cdot)$ is tractable.  By focusing on $T(\bm{U},\bm{\theta})$, we bypass the need to approximate the distribution of a point estimator or to invoke large‑sample asymptotics. For a detailed discussion of the choice of $T(\cdot)$ see \cite{xie2022reprosamplesmethodfinite}.
\section{Validating Proposed Method via Case Studies}

This section illustrates the versatility of our method across different ranking
scenarios. Section~\ref{quanteg} presents a quantile ranking example with unknown
population distributions. Sections~\ref{football} and~\ref{PL} address settings in which a
single observation informs multiple parameters: Section~\ref{football} considers a soccer
ranking problem requiring an algorithmic solution, and Section~\ref{PL} analyzes
partial rankings under the Plackett-Luce model where only top-ranked choices
are observed.

\subsection{Ranking Quantiles of Completely Unknown Distributions}
\label{quanteg}

We consider $K$ independent populations $\{\mathcal{P}_k\}_{k=1}^K$ with distribution functions 
$F_k(.)$, and aim to rank them according to their $\zeta$–quantiles  
$\theta_k^{(0)}$, defined by 
$F_k(\theta_k^{(0)})=\zeta$ for a fixed $\zeta\in(0,1)$. For each population $k$, observations $\{y^{\mathrm{obs}}_{ki}\}_{i=1}^{n_k}$ are independently drawn from $\mathcal{P}_k$. We impose no smoothness,  
shape, or parametric assumptions on~$F_k$.\\
\noindent\textbf{Oracle characterization:} For any variable $Y\sim F_{k}$, the indicator $\mathbb{I}(Y<\theta^{(0)}_k)$ follows a Bernoulli$(\zeta)$ distribution. We introduce latent errors $u^{\mathrm{rel}}_{ki}$ which are realizations from $\text{Bernoulli}(\zeta)$ distribution such that the oracle  
quantile is the solution to the implicit equation
\begin{equation}
\label{theta quant}
\theta^{(0)}_k = \arg\min_{\theta} \left\{ \sum_{i=1}^{n_k} \mathbb{I}(y^{obs}_{ki} - \theta < 0) - \sum_{i=1}^{n_k} u^{\mathrm{rel}}_{ki} \right\}
\end{equation}
Although the data cannot be written solely as a function of the parameter via (2), the generalized formulation in (3) still applies. 
Here the nuclear mapping $T(\bm{u},\bm{\theta})$ does not depend on $\bm{\theta}$ so we remove it from the notation. We write $
T_k(\bm{u}) = \sum_{i=1}^{n_k} u_{ki}
$ for each $k$ and the nuclear mapping
$T(\bm{u})^{K\times 1} = ( T_1(\bm{u}), T_2(\bm{u}),\ldots, T_K(\bm{u}) )$. Equation (\ref{theta quant})
implies that the solution $\theta^{(0)}_k$ is bracketed by the order statistics of the observed data as 
$ y^{obs,k}_{(T_k(u^{\mathrm{rel}}))} \le \theta^{(0)}_k \le y^{obs,k}_{(T_k(u^{\mathrm{rel}})+1)}
$ where $y^{obs,k}_{(r)}$ is the  $r^{th}$ sorted value within $\{y^{obs}_{ki}\}_{i=1}^{n_k}$ from $\mathcal{P}_k$. Then $y^{obs,k}_{(T_k(u^{\mathrm{rel}})+1)}< y^{obs,i}_{(T_i(u^{\mathrm{rel}}))}$, implies $\theta^{(0)}_k<\theta^{(0)}_i$ whereas  $y^{obs,i}_{(T_{i}(u^{\mathrm{rel}})+1)}<y^{obs,k}_{(T_k(u^{\mathrm{rel}}))}$ implies $\theta^{(0)}_k>\theta^{(0)}_i$.

\noindent{\textbf{Neighborhood sets and Borel region: }}
For any realization $\bm u\in\mathcal \mathcal{U}$, define the neighborhood sets $\mathcal N_k^{-}(\mathcal D^{\mathrm{obs}},\bm u)
=
\Bigl\{
i\neq k: 
y^{\mathrm{obs},k}_{(T_k(\bm u))}
>
y^{\mathrm{obs},i}_{(T_i(\bm u)+1)}
\Bigr\},$ $
\mathcal N_k^{+}(\mathcal D^{\mathrm{obs}},\bm u)
=
\Bigl\{
i\neq k:
y^{\mathrm{obs},k}_{(T_k(\bm u)+1)}
<
y^{\mathrm{obs},i}_{(T_i(\bm u))}
\Bigr\}.$
 As
$T_k(\bm u^{\mathrm{rel}})=\sum_{i=1}^{n_k} u^{\mathrm{rel}}_{ki}$ is a realization $ \mathrm{Binomial}(n_k,\zeta),$ we construct a marginal $(1-\alpha)^{1/K}$ Binomial confidence interval  
$[c_L^k,c_R^k]$ for each~$k$, defined as the shortest integer interval  
$(c^k_L, c^k_R) = \arg\min_{(i,j) \in \mathcal{A}_k} \; |j-i|$ with $\mathcal{A}_k = \left\{ (i,j) \bigg|\sum_{r=i}^{j} \binom{n_k}{r} \zeta^r (1-\zeta)^{n_k-r} \ge {(1-\alpha)}^{1/K} \right\}$.
The Borel set is
\begin{equation}\label{eq:borel-quantiles}
B_\alpha
=
\Bigl\{
T(\bm U)\;\big|\;
c_L^k\le T_k(\bm U)\le c_R^k,\ \forall k\in[K]
\Bigr\},
\end{equation}
and yields the preliminary confidence region  
$\Gamma_\alpha^\mathcal I(\mathcal D^{\mathrm{obs}})$ defined in  
\eqref{conf_set_final}.
Since $\bm u^{\mathrm{rel}}$ is unobserved, the neighborhood sets are evaluated for  
each artificial copy ${\bm u}^\star\in\mathcal U$.

\noindent{\textbf{Generated quantiles and candidate set:}} For any repro copy ${\bm u}^\star\in\mathcal V$, define the generated $\theta_k^\star
=
\arg\min_{\theta}
\Bigl\{
\sum_{i=1}^{n_k} I(y^{\mathrm{obs}}_{ki}<\theta)
-
\sum_{i=1}^{n_k} u^\star_{ki}
\Bigr\},$
$\widehat\theta_k^{\mathrm{obs}}
=
y^{\mathrm{obs},k}_{(\lceil n_k\zeta\rceil)}
=
\inf\{\theta:\widehat F_k^{\mathrm{obs}}(\theta)\ge\zeta\},$ the sample  
$\zeta$-quantile where $
\widehat F_k^{\mathrm{obs}}(\theta)
=\frac{1}{n_k}\sum_{i=1}^{n_k}
I(y^{\mathrm{obs}}_{ki}\le\theta),$ and the  discordance score comparing $\bm\theta^\star$ and  
$\widehat{\bm\theta}^{\mathrm{obs}}$ as
$\mathrm{Disc}(\mathcal D^{\mathrm{obs}},\bm\theta^\star)
=
\sum_{1\le i\ne j\le K}
I\!\Bigl(
(\widehat\theta_i^{\mathrm{obs}}
-\widehat\theta_j^{\mathrm{obs}})
(\theta_i^\star-\theta_j^\star)
<0
\Bigr).$ A repro sample is accepted into the candidate set  
$\mathcal C_{\mathcal{V}}(\mathcal{D}^{\mathrm{obs}})$ if  
$\mathrm{Disc}(\mathcal D^{\mathrm{obs}},\bm\theta^\star)<c$,  
in which case $\bm R^\star=\mathcal S(\bm\theta^\star)$ is retained.

\noindent{\textbf{Final confidence set: }}We collect all such $\bm R=(r_1,..,r_K)\in S_K,$ for which there exists any ${\bm u}^\star \in \mathcal{U}$ such that  
$T(\bm u^\star)$ lies in $B_\alpha$, and 
$\bigl|\mathcal N_k^{-}(\mathcal D^{\mathrm{obs}},\bm u^\star)\bigr|+1
\le
r_k
\le
K-
\bigl|\mathcal N_k^{+}(\mathcal D^{\mathrm{obs}},\bm u^\star)\bigr|.$ Then using \ref{alg:construction} and ~\ref{alg:candidate} we construct the
refined $(1-\alpha)$ confidence set $\tilde{\Gamma}_{\mathcal{V}_\alpha}(\mathcal D^{\mathrm{obs}}).$  

\subsection{Ranking in Competitive Sports via a Regression Model}
\label{football}
 Next we consider the problem of ranking $K$ sports teams according to latent ability
parameters $\theta_k$, where larger values of $\theta_k$ represent stronger teams.
Let $\bm Y\in\mathbb R^{n}$ denote observed game-level responses (e.g., goal
differences), and let $\bm X\in\mathbb R^{n\times K}$ be a fixed design matrix
encoding nonrandom covariates such as opponent indicators, match locations, or
other game characteristics. For a  noise vector $ \bm{U} \sim F_{\bm{U}}(.) $ and $\sigma$ an unknown scale parameter, we assume the linear regression model $\bm{Y} = \bm{X}\bm{\theta} + \sigma \bm{U}$. The sample realized version of the above model is given by 
$ \bm{y}^{\mathrm{obs}} = \bm{x}^{\mathrm{obs}}\bm{\theta}^{(0)} + \sigma \bm{u}^{\mathrm{rel}}.$\\
\noindent{\textbf{Oracle characterization:}}
By ordinary least squares, $\hat{\bm{\theta}}^{\mathrm{obs}} = (\hat{\theta}^{\mathrm{obs}}_1, \dots, \hat{\theta}^{\mathrm{obs}}_K) = (\bm{x}^{\mathrm{obs}\top}\bm{x}^{\mathrm{obs}})^{-1}\bm{x}^{\mathrm{obs}\top}\bm{y}^{\mathrm{obs}}$. Then, $\hat{\bm{\theta}}^{\mathrm{obs}} = \bm{\theta}^{(0)} + \sigma(\bm{x}^{\mathrm{obs}\top}\bm{x}^{\mathrm{obs}})^{-1}\bm{x}^{\mathrm{obs}\top}\bm{u}^{\mathrm{rel}}.$ 
Thus for $
A=(\bm{x}^{\mathrm{obs}\top}\bm{x}^{\mathrm{obs}})^{-1}\bm{x}^{\mathrm{obs}\top},$
where $A_k$ denotes the $k$th row of~$A$ we have 
$\theta^{(0)}_k
=\widehat\theta^{\mathrm{obs}}_k
-\sigma\,A_k\bm u^{\mathrm{rel}}.$\\
\noindent{\textbf{Neighborhood sets and Borel region:}} The sign of
$\theta_k^{(0)}-\theta_i^{(0)}$ is determined by the sign of
$\widehat\theta_k^{\mathrm{obs}}-\widehat\theta_i^{\mathrm{obs}}$
and the sign of $A_k\bm u^{\mathrm{rel}}-A_i\bm u^{\mathrm{rel}}$, that is $\{A_k\bm u^{\mathrm{rel}}<A_i\bm u^{\mathrm{rel}}, 
\widehat\theta_k^{\mathrm{obs}}>\widehat\theta_i^{\mathrm{obs}}\}\subseteq
\{\theta_k^{(0)}>\theta_i^{(0)}
\},$
and similarly for the reversed inequality.  
Thus for any noise vector $\bm u\in\mathcal V$, define
\begin{align}\label{eq:linear-neighborhoods}
\mathcal N_k^{-}(\mathcal D^{\mathrm{obs}},\bm u)
&=
\bigl\{
i\neq k:\ 
A_k\bm u<A_i\bm u,\ 
\widehat\theta_k^{\mathrm{obs}}>\widehat\theta_i^{\mathrm{obs}}
\bigr\}\\
\mathcal N_k^{+}(\mathcal D^{\mathrm{obs}},\bm u)
&=
\bigl\{
i\neq k:\ 
A_k\bm u>A_i\bm u,\ 
\widehat\theta_k^{\mathrm{obs}}<\widehat\theta_i^{\mathrm{obs}}
\bigr\}
\end{align}
Here the nuclear mapping is $T(\bm u)=\bm u$ itself.  
To control sampling variability of~$\bm u^{\mathrm{rel}}$, we construct the
coordinatewise region $B_\alpha
=
\Bigl\{
\bm U:\ 
c_L^i\le U_i\le c_R^i,\quad i=1,\dots,n
\Bigr\},$
where $[c_L^i,c_R^i]$ is a marginal $(1-\alpha)^{1/n}$ interval for~$u_i$ under
$F_{\bm U}$. Laplace errors are frequently used in competitive-sports applications (e.g.,
soccer goal differences) due to heavy-tailed error patterns and sharp central
peaks.  Such non-Gaussian noise leads to the absence of closed-form estimators
for $\sigma,$ hence our inference is constructed without directly estimating
$\sigma$ from the data. \\
\noindent{\textbf{Generated $\bm \theta^\star$ and Candidate set:}} To generate repro parameters, we draw $ {\bm u}^\star $ from $F_{\bm{U}}(\cdot) $,
For a scale $\sigma^\star$, we rewrite the model 
$\bm y^{\mathrm{obs}}
=\bm x^{\mathrm{obs}}\bm\theta^\star
+\sigma^\star {\bm u}^\star,$ as 
$\bm y^{\mathrm{adj}}
=\bm y^{\mathrm{obs}}-\sigma^\star{\bm u}^\star
=\bm x^{\mathrm{obs}}\bm\theta^\star.$
Given $\sigma^\star$, the corresponding repro-sample parameter satisfies $$\bm\theta^\star(\sigma^\star)=
(\bm x^{\mathrm{obs}\top}\bm x^{\mathrm{obs}})^{-1}
\bm x^{\mathrm{obs}\top}
\bigl(\bm y^{\mathrm{obs}}-\sigma^\star{\bm u}^\star\bigr).$$ The scale $\sigma^\star$ minimizes the residual sum of squares:
$\sigma^\star
=
\arg\min_{\sigma>0}
\bigl\|
\bm y^{\mathrm{obs}}
-
\sigma{\bm u}^\star
-
\bm x^{\mathrm{obs}}
\bm\theta^\star(\sigma)
\bigr\|^2,$ which we solve by Brent’s method.  
Iteratively solving the above equations till convergence yields the repro-sample parameter~$\bm\theta^\star$.
For each repro draw $\bm\theta^\star$, define the discordance count 
$\mathrm{Disc}(\mathcal D^{\mathrm{obs}},\bm\theta^\star)
=
\sum_{1\le i\ne j\le K}
\mathbb I\!\Bigl(
(A_i\bm u^{\mathrm{rel}}-A_k\bm u^{\mathrm{rel}})
(\theta_i^\star-\theta_k^\star)
<0
\Bigr).$
If the discordance count is less than c, include $\bm R^\star=\mathcal S(\bm\theta^\star)$ in the
candidate set $\mathcal C_{\mathcal{V}}(\mathcal{D}^{\mathrm{obs}})$. Using Algorithms \ref{alg:construction} and ~\ref{alg:candidate} we obtain $\tilde{\Gamma}_{\mathcal{V}_\alpha}(\mathcal{D}^{\mathrm{obs}}).$
            
\subsection{Ranking Plackett--Luce Parameters for Top-Choice Data}
\label{PL}

We now consider ranking items under the Plackett--Luce (PL) model, a standard
framework for analyzing top-choice or partial ranking data.  Each item
$k\in[K]$ possesses a positive worth parameter $\theta_k>0$, where larger
$\theta_k$ indicates a higher chance of being chosen.  In each trial $t$, a
subset of items
$S_t^{\mathrm{obs}}
=\{j^t_1<j^t_2<\cdots<j^t_M\}\subseteq[K]$
is presented, from which a single item is observed as the top choice as given in \cite{fan2024ranking}. Under
the PL model, item $j^t_m\in S_t^{\mathrm{obs}}$ is selected with probability
$\mathbb P(j^t_m \text{ chosen}\mid S_t^{\mathrm{obs}})
=
\frac{\theta_{j^t_m}}{\sum_{k\in S_t^{\mathrm{obs}}}\theta_k}.$
This setting arises in applications such as peer review, consumer preference
surveys, and subset-wise recommendation systems.  Suppose each $M$-subset is
repeated $L$ times, yielding $T=\binom{K}{M}L$ observed trials.

\noindent{\textbf{Oracle characterization via quadratic programming:}}
Let $\bm u^{\mathrm{rel}}=(u^{\mathrm{rel}}_1,\ldots,u^{\mathrm{rel}}_T)$ denote
unobserved uniform noise $u^{\mathrm{rel}}_t$ from $\mathrm{Unif}(0,1)$ determining the
selected item in each trial. If   
$S^{\mathrm{obs}}_t=\{j^t_1,\ldots,j^t_M\}$ and item $j^t_m$ is chosen at
trial~$t$, then under the PL generative mechanism,
\begin{equation}\label{eq:PL-ineq}
\sum_{r=1}^{m-1}\theta^{(0)}_{j^t_r}
<
u^{\mathrm{rel}}_t
\sum_{k\in S_t^{\mathrm{obs}}}\theta^{(0)}_k
\;\le\;
\sum_{r=1}^m \theta^{(0)}_{j^t_r}.
\end{equation}
For each trial, the inequalities~\eqref{eq:PL-ineq} can be written as linear
constraints.  Let $G^{\mathrm{rel}}\in\mathbb R^{2T\times K}$ contain the $2T$
rows constructed from~\eqref{eq:PL-ineq}.  For trial $t$ with top choice
$j^t_m$, the two constraint rows $G^{\mathrm{rel}}_{2t-1}$ and
$G^{\mathrm{rel}}_{2t}$ are
{\small
\[
(G^{\mathrm{rel}}_{2t-1,k})=
\begin{cases}
1-u^{\mathrm{rel}}_t,& k\in\{j^t_1,\dots,j^t_{m-1}\},\\
-\,u^{\mathrm{rel}}_t,& k\in S_t^{\mathrm{obs}}\setminus\{j^t_1,\dots,j^t_{m-1}\},\\
0,&\text{otherwise},
\end{cases}
\qquad
(G^{\mathrm{rel}}_{2t,k})=
\begin{cases}
u^{\mathrm{rel}}_t-1,& k\in\{j^t_1,\dots,j^t_m\},\\
u^{\mathrm{rel}}_t,& k\in S_t^{\mathrm{obs}}\setminus\{j^t_1,\dots,j^t_m\},\\
0,&\text{otherwise}.
\end{cases}
\]
}
We define the oracle worth vector $\bm\theta^{(0)}$ as the solution of
\begin{equation}\label{eq:PL-oracle-QP}
\min_{\bm\theta^{(0)}\ge0}\ \|\bm\theta^{(0)}\|_2^2
\quad\text{subject to}\quad
G^{\mathrm{rel}}\bm\theta^{(0)}\le0,\qquad 
\sum_{k=1}^K\theta^{(0)}_k=1.
\end{equation}

\noindent{\textbf{Neighborhood sets and Borel region:}}
Fix a subset of items $S=\{j_1<j_2<j_3\}$, and define
$\mathcal T_S=\{\, t\in\{1,\dots,T\} : S_t^{\mathrm{obs}}=S \,\},$
$L \;=\; |\mathcal T_S|.$ For each $t\in\mathcal T_S$, let $u^{\mathrm{rel}}_t\sim\mathrm{Unif}(0,1)$ denote the latent noise used to generate the top choice under the Plackett--Luce mechanism.  Let
$u^{\mathrm{rel}}_{(1)}<\cdots<u^{\mathrm{rel}}_{(L)}$
be the corresponding order statistics. Let $y_i^{\mathrm{obs}}$ denote the number of times item $j_i$ is selected as the top choice among these $L$ trials, and write
$\bm y^{\mathrm{obs}}_S
=
\bigl(y_1^{\mathrm{obs}},y_2^{\mathrm{obs}},y_3^{\mathrm{obs}}\bigr),$
$y_1^{\mathrm{obs}}+y_2^{\mathrm{obs}}+y_3^{\mathrm{obs}} = L.$
From the PL inequalities~\eqref{eq:PL-ineq}, we obtain the ratio bounds
\[
\frac{u^{\mathrm{rel}}_{(y_1^{\mathrm{obs}})} - u^{\mathrm{rel}}_{(1)}}
     {u^{\mathrm{rel}}_{(y_1^{\mathrm{obs}}+y_2^{\mathrm{obs}}+1)} -
      u^{\mathrm{rel}}_{(y_1^{\mathrm{obs}})}}
\;<\;
\frac{\theta^{(0)}_{j_1}}{\theta^{(0)}_{j_2}}
\;<\;
\frac{u^{\mathrm{rel}}_{(y_1^{\mathrm{obs}}+1)}}
     {u^{\mathrm{rel}}_{(y_1^{\mathrm{obs}}+y_2^{\mathrm{obs}})} -
      u^{\mathrm{rel}}_{(y_1^{\mathrm{obs}}+1)}},
\]
\[
\frac{u^{\mathrm{rel}}_{(y_1^{\mathrm{obs}})} - u^{\mathrm{rel}}_{(1)}}
     {1 - u^{\mathrm{rel}}_{(y_1^{\mathrm{obs}}+y_2^{\mathrm{obs}})}}
\;<\;
\frac{\theta^{(0)}_{j_1}}{\theta^{(0)}_{j_3}}
\;<\;
\frac{u^{\mathrm{rel}}_{(y_1^{\mathrm{obs}}+1)}}
     {u^{\mathrm{rel}}_{(L)} -
      u^{\mathrm{rel}}_{(y_1^{\mathrm{obs}}+y_2^{\mathrm{obs}}+1)}},
\]
\[
\frac{u^{\mathrm{rel}}_{(y_1^{\mathrm{obs}}+y_2^{\mathrm{obs}})} -
      u^{\mathrm{rel}}_{(y_1^{\mathrm{obs}}+1)}}
     {1 -
      u^{\mathrm{rel}}_{(y_1^{\mathrm{obs}}+y_2^{\mathrm{obs}})}}
\;<\;
\frac{\theta^{(0)}_{j_2}}{\theta^{(0)}_{j_3}}
\;<\;
\frac{u^{\mathrm{rel}}_{(y_1^{\mathrm{obs}}+y_2^{\mathrm{obs}}+1)} -
      u^{\mathrm{rel}}_{(y_1^{\mathrm{obs}})}}
     {u^{\mathrm{rel}}_{(L)} -
      u^{\mathrm{rel}}_{(y_1^{\mathrm{obs}}+y_2^{\mathrm{obs}}+1)}}.
\]
These inequalities provide partial orderings of the parameters. For instance, $\frac{u^{\mathrm{rel}}_{(y_1^{\mathrm{obs}}+1)}}{u^{\mathrm{rel}}_{(y_1^{\mathrm{obs}}+y_2^{\mathrm{obs}})} - u^{\mathrm{rel}}_{(y_1^{\mathrm{obs}}+1)}} < 1$ implies $\theta^{(0)}_{j_1} < \theta^{(0)}_{j_2},$
while
$\frac{u^{\mathrm{rel}}_{(y_1^{\mathrm{obs}})} - u^{\mathrm{rel}}_{(1)}}
{u^{\mathrm{rel}}_{(y_1^{\mathrm{obs}}+y_2^{\mathrm{obs}}+1)} -
      u^{\mathrm{rel}}_{(y_1^{\mathrm{obs}})}} > 1$
implies 
$\theta^{(0)}_{j_1} > \theta^{(0)}_{j_2}.$ To study the rank of a given item $k$, we collect all trials $t$ for which $k\in S_t^{\mathrm{obs}}$.  For each such trial, writing $S_t^{\mathrm{obs}}=\{j_1^t<j_2^t<j_3^t\}$, the corresponding inequalities give pairwise information for the ordered pairs $(k,i)\subset S_t^{\mathrm{obs}}$. For each ordered pair $(k,i)\subset S$, we define indicator functions
$I_S^{k<i}(\bm u^{\mathrm{rel}})$ and $I_S^{k>i}(\bm u^{\mathrm{rel}})$, which record whether the above constraints imply $\theta^{(0)}_k<\theta^{(0)}_i$ or $\theta^{(0)}_k>\theta^{(0)}_i$ as 
{\scriptsize
\[
\begin{aligned}
I_S^{k<i}(\bm{u}^{\mathrm{rel}})
&=
\begin{cases}
1,&
(k,i)=(j_1,j_2),
\dfrac{u^{\mathrm{rel}}_{(y^{\mathrm{obs}}_1+1)}}
      {u^{\mathrm{rel}}_{(y^{\mathrm{obs}}_1+y^{\mathrm{obs}}_2)}
       -u^{\mathrm{rel}}_{(y^{\mathrm{obs}}_1+1)}}<1,\\
1,&
(k,i)=(j_1,j_3),
\dfrac{u^{\mathrm{rel}}_{(y^{\mathrm{obs}}_1+1)}}
      {u^{\mathrm{rel}}_{(L)}
       -u^{\mathrm{rel}}_{(y^{\mathrm{obs}}_1+y^{\mathrm{obs}}_2+1)}}<1,\\
1,&
(k,i)=(j_2,j_3),
\dfrac{u^{\mathrm{rel}}_{(y^{\mathrm{obs}}_1+y^{\mathrm{obs}}_2+1)}
      -u^{\mathrm{rel}}_{(y^{\mathrm{obs}}_1)}}
      {u^{\mathrm{rel}}_{(L)}
       -u^{\mathrm{rel}}_{(y^{\mathrm{obs}}_1+y^{\mathrm{obs}}_2+1)}}<1,\\
0,&\text{otherwise},
\end{cases}
\,
I_S^{k>i}(u^{\mathrm{rel}})
=
\begin{cases}
1,&
(k,i)=(j_1,j_2),
\dfrac{u^{\mathrm{rel}}_{(y^{\mathrm{obs}}_1)}
      -u^{\mathrm{rel}}_{(1)}}
      {u^{\mathrm{rel}}_{(y^{\mathrm{obs}}_1+y^{\mathrm{obs}}_2+1)}
       -u^{\mathrm{rel}}_{(y^{\mathrm{obs}}_1)}}>1,\\
1,&
(k,i)=(j_1,j_3),
\dfrac{u^{\mathrm{rel}}_{(y^{\mathrm{obs}}_1)}
      -u^{\mathrm{rel}}_{(1)}}
      {1
       -u^{\mathrm{rel}}_{(y^{\mathrm{obs}}_1+y^{\mathrm{obs}}_2)}}>1,\\
1,&
(k,i)=(j_2,j_3),
\dfrac{u^{\mathrm{rel}}_{(y^{\mathrm{obs}}_1+y^{\mathrm{obs}}_2)}
      -u^{\mathrm{rel}}_{(y^{\mathrm{obs}}_1+1)}}
      {1
       -u^{\mathrm{rel}}_{(y^{\mathrm{obs}}_1+y^{\mathrm{obs}}_2)}}>1,\\
0,&\text{otherwise}.
\end{cases}
\end{aligned}
\]
\normalsize
For each item $k\in[K]$, we define
$\mathcal N_k^{-}(\mathcal D^{\mathrm{obs}},\mathbf u)
=
\left\{
\, i \neq k :
\exists\, S \ni k,i \ \text{such that}\ 
I_{S}^{\,k>i}(\mathbf u)=1
\right\},$ and
$\mathcal N_k^{+}(\mathcal D^{\mathrm{obs}},\mathbf u)
=
\left\{
\, i \neq k :
\exists\, S \ni k,i \ \text{such that}\ 
I_{S}^{\,k<i}(\mathbf u)=1
\right\}$ where each set records items forced to be below or above $k$. We write $T(\bm u)=\bm u$, and denote its order statistics by
$u_{(1)}<\cdots<u_{(T)}$.  For each $t=1,\dots,T$, choose
$c_L^t=F_{U_{(t)}}^{-1}(\alpha/2),$ 
$c_R^t=F_{U_{(t)}}^{-1}(1-\alpha/2),$ and define $B_\alpha
=
\bigl\{
\bm U\in[0,1]^T:\ 
c_L^t<U_{(t)}<c_R^t,t=1,\dots,T
\bigr\}.$\\
\noindent{\textbf{Generated parameter and candidate set:}} As $\bm u^{\mathrm{rel}}$ is unobserved, we draw
$\bm u^\star$ from $\mathrm{Unif}(0,1)^T$ and construct a constraint matrix
$G^\star$ identical to $G^{\mathrm{rel}}$ but with $u^{\mathrm{rel}}_t$ replaced
by $u^\star_t$.  For each repro draw, compute $\bm\theta^\star
=\arg\min_{\bm\theta\ge0}\ \|\bm\theta\|_2^2,$ such that
$G^\star\bm\theta\le0,$ $\sum_{k=1}^K\theta_k=1.$ Each $\bm\theta^\star$ yields a repro ranking
$\bm R^\star=\mathcal S(\bm\theta^\star)$.  We use the discordance
criterion
$\mathrm{Disc}(\mathcal D^{\mathrm{obs}},\bm\theta^\star)$
based on the PL estimator $\widehat{\bm\theta}^{\mathrm{obs}}$ given in \cite{fan2024ranking}} after correcting with a log factor since their probabilities are proportional to $e^\theta_k.$ If $\mathrm{Disc}(\mathcal D^{\mathrm{obs}},\bm\theta^\star)<c$, then $\bm R^\star$ is placed in the candidate set
$\mathcal C_{\mathcal{V}}(\mathcal{D}^{\mathrm{obs}})$. Algorithm \eqref{alg:construction} and \eqref{alg:candidate} then produce
$\tilde{\Gamma}_{\mathcal{V}_\alpha}^{\mathcal I}(\mathcal D^{\mathrm{obs}})$.

\section{Numerical Illustrations with Real-World Data}\label{numerical}

\subsection{Quantile-Based Ranking of National Wealth Distributions}
We analyzed wealth data from Forbes’ 2024 \textit{World’s Billionaires} report, using Section \ref{quanteg} theory, restricting attention to all individuals with estimated net worths exceeding \$5 million from the United States, Germany, Russia, India, and China. For each country $k$, the $\zeta$-quantile was estimated by the empirical order statistic
$\widehat{\theta}_k
= y^{\mathrm{obs},k}_{(\lceil n_k \zeta\rceil)},$
and we constructed rank confidence sets for $\zeta=0.5$ and $\zeta=0.75.$ Our  method was implemented using Bernoulli$(\zeta)$ latent noise, $B=2000$ resamples, discordance budget 
$c=\lfloor p^\ast K_{\mathrm{pairs}}\rfloor$ with $p^\ast=0.20$. Repro samples passing both the Borel and discordance filters were retained, and the intersection of their induced ranks formed the joint $(1-\alpha)$ confidence set. The resulting rank intervals (Table~\ref{tab:quantile-rank-ci}) are nontrivial and always contain the empirical ranks. At the upper quartile, Russia and India are sharply identified as the top two countries, whereas the United States, Germany, and China exhibit broader but still informative intervals. For comparison, we implemented the simultaneous bootstrap procedure. For each contrast $\theta_k - \theta_i$ we constructed Bonferroni–adjusted simultaneous confidence intervals
$[c^L_{kj},\,c^R_{kj}]$ such that $\mathbb{P}\!\left(
c^L_{kj} \le \widehat{\theta}_k - \widehat{\theta}_i \le c^R_{kj}
\;\;\forall\,k,j
\right)\ge 1-\alpha,$ based on the difference of sample quantiles $\widehat{\theta}_k - \widehat{\theta}_i.$ A country $k$ is then deemed certainly ahead of country $j$ if $c^L_{kj}>0$, and certainly behind $j$ if $c^R_{kj}<0$. This yields
$\mathcal N_k^{-}=\{j\ne k:\; c^L_{kj}>0\},$
$\mathcal N_k^{+}=\{j\ne k:\; c^R_{kj}<0\}.$
Following \cite{fan2024ranking}, the simultaneous rank confidence interval for country $k$ is
$1+|\mathcal N_k^{-}|
\;\le\;
r_k
\;\le\;
K-|\mathcal N_k^{+}|.$
For both $\zeta=0.5$ and $\zeta=0.75$, all bootstrap intervals overlapped zero for every pair $(k,j)$, producing the trivial rank set $[1,5]$ for all countries. In contrast, the repro-sampling method produced sharp, interpretable, and finite-sample valid rank sets (Table~\ref{tab:quantile-rank-ci}), offering substantially greater discriminatory power in this heavy-tailed finite sample setting.
\begin{table}[ht]
\centering
\caption{Sample quantiles and rank confidence intervals for the Forbes dataset.}
\label{tab:quantile-rank-ci}
\small
\begin{tabular}{lccccccc}
\toprule
& \multicolumn{3}{c}{$\zeta = 0.5$ (Median)} 
& \multicolumn{3}{c}{$\zeta = 0.75$} 
& \multicolumn{1}{c}{$\zeta = 0.5$, $\zeta = 0.75$ } \\
\cmidrule(lr){2-4}
\cmidrule(lr){5-7}
\cmidrule(lr){8-8}
Country 
& Sample quantile
& Rank 
& Repro CI 
& Sample quantile
& Rank 
& Repro CI 
& Bootstrap CI \\
\midrule
US      & 8.1 & 4    & [3, 5] & 12.4 & 3   & [3, 4] & [1, 5] \\
Germany & 7.9 & 5    & [4, 5] & 12.1 & 4   & [3, 5] & [1, 5] \\
Russia  & 9.8 & 1    & [1, 2] & 21.1 & 1   & [1, 1] & [1, 5] \\
India   & 8.5 & 2--3 & [2, 4] & 17.6 & 2   & [2, 2] & [1, 5] \\
China   & 8.5 & 2--3 & [1, 4] & 11.6 & 5   & [4, 5] & [1, 5] \\
\bottomrule
\end{tabular}
\end{table}

\normalsize
\subsection{Ranking EPL Teams from Pairwise Score Differences}
We apply the procedure of Section~\ref{football} to the 2023--2024 English Premier League (EPL) dataset.
For each match $i=1,\dots,n$, of $n$ total matches the observed goal difference is modeled as $y^{\mathrm{obs}}_{i}
= \theta^{(0)}_{h(i)} - \theta^{(0)}_{a(i)} + \delta^{(0)} + \sigma^{(0)} u^{\mathrm{rel}}_{i},$ where $h(i)$ and $a(i)$ denote the home and away teams, $\delta^{(0)}$ is the home-field intercept. Here the game-level noise $u^{\mathrm{rel}}_i$ is a realization from $\mathrm{Laplace}(0,1).$ Stacking all matches yields the realized linear model
$\bm y^{\mathrm{obs}}
= \bm x^{\mathrm{obs}} \bm\theta^{(0)} + \delta^{(0)} \mathbf{1}_n + \sigma^{(0)} \bm u^{\mathrm{rel}},$
with the design matrix $ x^{\mathrm{obs}}$ encoding $+1$ for the home team and $-1$ for the away team.
We impose the identifiability constraint $\sum_{k=1}^K \theta_k^{(0)} = 0$, and compute the least-squares estimator $\widehat{\bm\theta}$.
The matrix $A = ( \bm x^{\mathrm{obs}^\top}  \bm x^{\mathrm{obs}})^{+}  \bm x^{\mathrm{obs}^\top},$
needed for the discordance statistic, is computed using the constraint-adjusted generalized inverse of $\bm x^{\mathrm{obs}^\top}  \bm x^{\mathrm{obs}}$.
The Borel set $B_\alpha$ from Section~\ref{football} is estimated by Monte Carlo simulation so that $\mathbb{P}_{\bm{U}}(B_\alpha) \approx 1-\alpha$ with $\alpha=0.05$. We retain a artificial copy $\bm u^{\ast(b)}$ if $\mathrm{Disc}(\mathcal{D}^{\mathrm{obs}}_n,\bm \theta^{(b)}) < c,$ $c=420,$ taking $p^\star=0.02.$ Among $2000$ total repro samples, only repro draws satisfying both the Borel screening and the discordance threshold contribute to the final set. Table~\ref{tab:epl-ci-theta} presents the resulting rank intervals together with the official EPL points.
We observe that the teams with fewer decisive performance gaps have larger intervals or show greater uncertainty. 
Importantly, the repro-sample rank intervals align closely with the official standings based on goal difference, demonstrating that the procedure captures the competitive structure of the league.

\begin{table}[ht]
\centering
\caption{Our rank confidence sets with traditional goal-difference rankings for EPL 2023–24 }
\label{tab:epl-ci-theta}
\setlength{\tabcolsep}{5pt}
\renewcommand{\arraystretch}{1.2}
\begin{tabular}{l c c c  l c c c}
\toprule
\textbf{Team} & \textbf{Rank CI} & \textbf{GD} & \textbf{GD Rank} 
& \textbf{Team} & \textbf{Rank CI} & \textbf{GD} & \textbf{GD Rank} \\
\midrule

Man City       & [1,2]    & 62   & 1.5  & Brighton       & [10,13] & -7   & 11 \\
Arsenal        & [1,2]    & 62   & 1.5  & Bournemouth    & [12,16] & -13  & 12 \\
Liverpool      & [3,3]    & 45   & 3    & Fulham         & [10,12] & -6   & 10 \\
Aston Villa    & [5,7]    & 15   & 5    & Wolves         & [14,17] & -15  & 14 \\
Spurs          & [6,7]    & 13   & 7    & Everton        & [12,16] & -11  & 13 \\
Chelsea        & [5,7]    & 14   & 6    & Brentford      & [10,14] & -9   & 12 \\
Newcastle      & [4,4]    & 23   & 4    & Nott'm Forest  & [16,17] & -18  & 17 \\
Man Utd        & [8,9]    & -1   & 8.5  & Luton          & [18,19] & -33  & 18 \\
West Ham       & [13,16]  & -14  & 15   & Burnley        & [18,19] & -37  & 19 \\
Crystal Palace & [8,9]    & -1   & 8.5  & Sheffield Utd  & [20,20] & -69  & 20 \\
\bottomrule
\end{tabular}
\end{table}

\subsection{Ranking Jokes Using the Plackett–Luce Model}

As an illustration of Section~\ref{PL} methodology, we rank jokes using the PL model from the Jester dataset \textit{(https://goldberg.berkeley.edu/jester-data)}. We analyze a subset of 10 jokes evaluated by 10 users and 80 users, each choosing their favorite from subsets of 3 jokes. All possible triplets ($\binom{10}{3}=120$) are evaluated by each user, yielding $1200$ and $9600$ top-choice observations. To construct confidence sets for joke ranks, we generated 2000 artificial noise samples, using Dirichlet bands based on uniform order statistics. Here we do not use a candidate set, or trivially take $c$ to be $K(K-1)$ which is $90$.
For comparison, we applied the algorithm from \cite{fan2024ranking} to the same dataset, estimating $\hat{\bm{\theta}}.$ We performed a bootstrap on 2000 samples to derive simultaneous confidence intervals for parameter differences and joke ranks, obtaining the simultaneous critical value $\zeta_{0.95}$. Our resulting confidence sets are substantially narrower than those from \cite{fan2024ranking} for smaller dataset with $10$ users or replicated comparisons for any set of three jokes. For the larger dataset our results are comparable highlighting our method's effectiveness for finite samples without relying on asymptotic assumptions. Moreover our method almost always covers the MLE estimate of the joke based on \cite{fan2024ranking} in the interval.

\begin{table}[ht]
\centering
\small
\caption{Rank confidence intervals for Jester data using our method and \cite{fan2024ranking}'s method}
\label{tab:jester-rank-ci}
\begin{tabular}{ccccccccc}
\toprule
& \multicolumn{4}{c}{Repetitions L = 10} 
& \multicolumn{4}{c}{Repetitions L = 80} \\
\cmidrule(lr){2-5}
\cmidrule(lr){6-9}
Joke ID 
& $\hat{\theta}_{mle}$ 
& Rank$_{10}$
& Repro CI 
& Bootstrap CI
& $\hat{\theta}_{mle}$ 
& Rank$_{80}$
& Repro CI 
& Bootstrap CI
\\
\midrule
Joke 5  & 0.11775 & 2 & [2,7]  & [2, 8]   & 0.12717 & 2 & [2,2]  & [2, 4] \\
Joke 7  & 0.11164 & 4 & [2,6]  & [2, 9]   & 0.10599 & 4 & [4,5]  & [2,5] \\
Joke 8  & 0.08113 & 6 & [4,10]  & [2, 10]   & 0.08410 & 6 & [6,6]  & [5,7] \\
Joke 13 & 0.07626 & 7 & [4,9]  & [2, 10]   & 0.06539 & 9 & [8,9]  & [7, 9] \\
Joke 15 & 0.07530 & 8 & [5,9] & [2, 10]   & 0.06899 & 7 & [6,8]  & [6, 9] \\
Joke 16 & 0.04750 & 10 & [6,10] & [6, 10]  & 0.04384 & 10 & [10,10] & [10,10] \\
Joke 17 & 0.06516 & 9  & [5,10] & [4, 10]  & 0.06662 & 8  & [7,9]   & [7, 9] \\
Joke 18 & 0.09031 & 5  & [2,7]  & [2, 9]  & 0.09987 & 5  & [4,6]   & [3, 6] \\
Joke 19 & 0.11406 & 3  & [1,8]  & [2, 8]  & 0.11616 & 3  & [3,3]   & [2, 5] \\
Joke 21 & 0.22089 & 1  & [1,2]  & [1, 1]  & 0.22187 & 1  & [1,1]   & [1, 1] \\
\bottomrule
\end{tabular}
\end{table}

\subsection{Ranking Hospitals with Unequal Variances}

To demonstrate the methodology in Example 2.1, we analyze data from the National Committee for Quality Assurance (NCQA) Quality Compass Report on blood glucose (A1c) control among diabetic patients across 78 Veterans Health Administration (VHA) hospitals in the United States \citep{Miller}. The objective is to rank hospitals based on the latent log-odds of good A1c control. The observed statistic for each hospital is $ y_k^{\text{obs}} = \log\left( \frac{\hat p^{\mathrm{obs}}_k}{1 - \hat p^{\mathrm{obs}}_k} \right) $, where $\hat p^{\mathrm{obs}}_k$ is the estimated proportion of well-controlled cases. Under large-sample theory, we model
$y_k^{\text{obs}} = \log\left( \frac{p_k}{1 - p_k} \right) + \sigma_k u^{rel}_k,$
with $u^{\mathrm{rel}}_k$ from $ N(0,1)$ and $\sigma_k$ estimated as $\sqrt{1 / \hat p^{\mathrm{obs}}_k + 1 / (1 - \hat p_k)}.$ Using 1000 artificial noise copies, $p^*=0.10,$ we construct marginal confidence sets for the ranks of the true log-odds parameters $\theta_k^{(0)} = \log(p_k / (1 - p_k))$. Figure~\ref{fig:hospital}(a) shows the estimated log-odds and associated $95\%$ confidence intervals, which frequently overlap, highlighting the difficulty of directly inferring ranks from point estimates alone. To address this, we construct a confidence set using the neighborhood-based procedure described in Example 2.1. The resulting rank confidence set, shown in Figure \ref{fig:hospital}(b), provides a valid finite-sample inference for the discrete ranks of hospitals. Our results are comparable to those of \citet{xie2009}, but offer an improvement by providing exact confidence sets for exact discrete rank parameters rather than for smoothed quantities. 

\begin{figure}[htbp]
\centering
  \begin{subfigure}{.5\textwidth}
    \centering
    \includegraphics[width=1\linewidth]{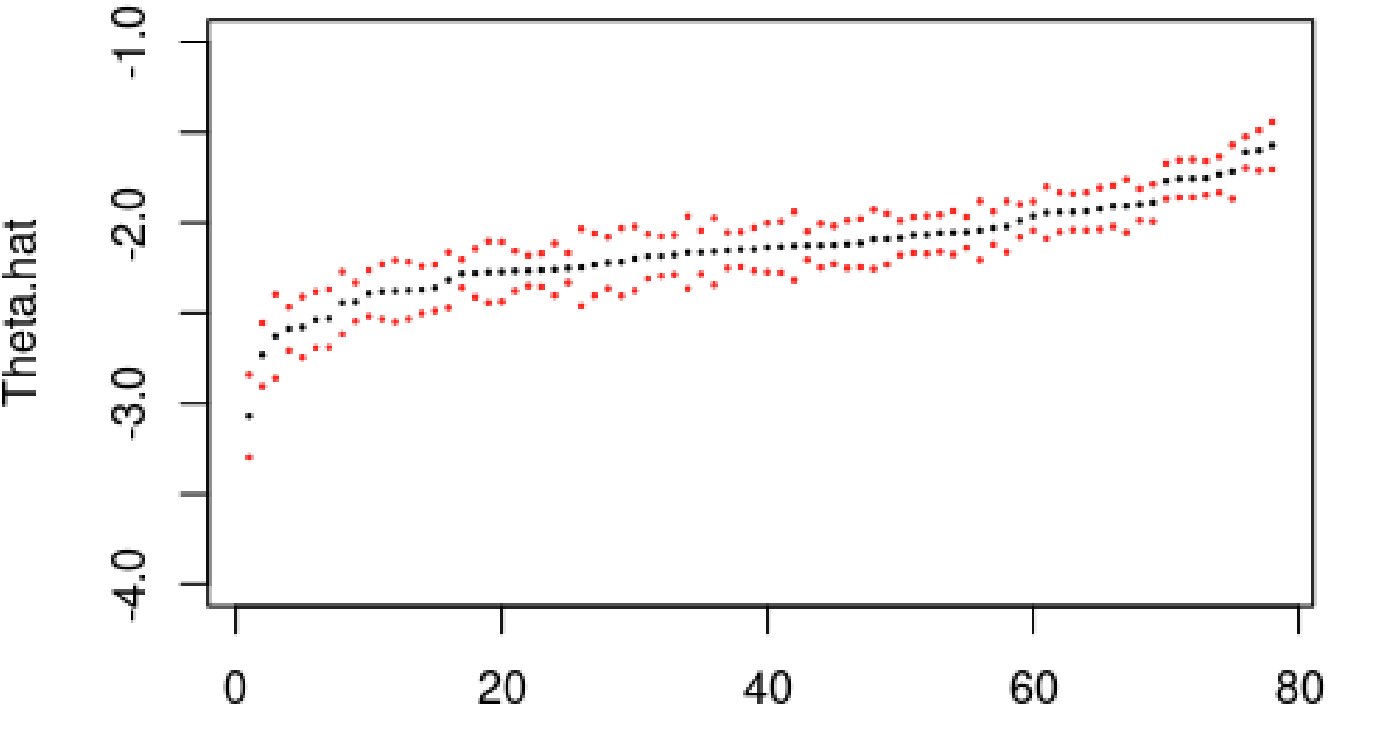}
    \caption{Estimated log-odds with 95\% confidence intervals.}
    \label{fig:a}
  \end{subfigure}
  \hspace{0.01\textwidth}
  \begin{subfigure}{.42\textwidth}
    \centering
    \includegraphics[width=1.1\linewidth]{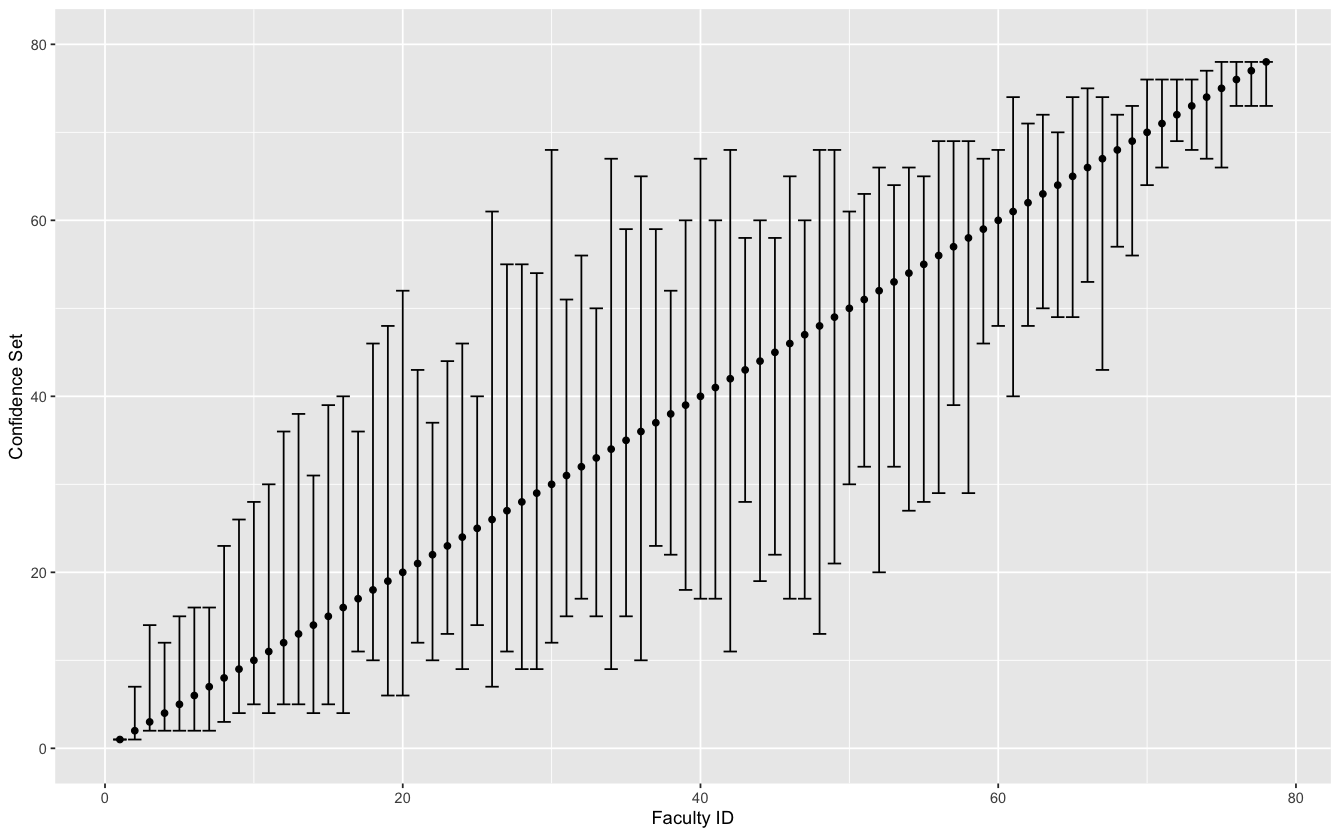}
    \caption{95\% simultaneous confidence set for VHA facilities}
    \label{fig:b}
  \end{subfigure}
  \caption{Confidence intervals and rank sets for A1c control across 78 hospitals. Each black point denotes the observed rank of the estimated log-odds.}
  \label{fig:hospital}
\end{figure}

\section{Simulation Study}
This section examines the finite-sample behaviour of the proposed rank confidence sets across three contrasting regimes: (i) ranking of quantiles under unknown distributions (ii) heavy-tailed designs motivated by the football application in Section~\ref{football}, (iii) unequal-variance Gaussian models reflecting the structure of the hospital dataset in Example~2.1. Our focus is on marginal and joint coverage of the true ranks, together with assessments of interval width and the effective size of candidate set with the choice of c in the Gaussian case.

\subsection{Quantile-Based Rank Coverage}

We first consider a quantile-regression setting in which the underlying populations follow a lognormal distribution. 
Specifically, for each $k=1,\dots,16$ we generated $y_{ik} $ from
$\mathrm{Lognormal}(\mu_k,\sigma^{(0)}),$ $\mu_k^{(0)}$ lies in $(11,14)$ and $\sigma^{(0)}$ in $
(0.1,0.4)$ so that the true distributions differ systematically in location and scale.
Each population was sampled with $n=100$ observations, and across $1000$ Monte Carlo replications we estimated and ranked the $75$th percentiles (oracle quantiles) of the lognormal populations. Repro-based marginal $95\%$ rank confidence sets were constructed using $1500$ artificial perturbations per replication, where each perturbation was drawn from $\mathrm{Binomial}(n,0.75)$ to mimic the score structure of the empirical quantile estimator. 
The candidate-set threshold was fixed at $p^\ast = 0.20$, discarding perturbations whose pairwise orderings deviated excessively from the empirical ordering of the estimated quantiles. In addition to marginal performance, we computed the global joint coverage across replications. 
The joint coverage of the repro-based $95\%$ rank confidence sets was $0.976.$
Table~\ref{tab:quantile_joint} reports the resulting marginal coverage probabilities and the corresponding mean and standard deviation of interval lengths. The results demonstrate that the repro-sampling method maintains nominal or  super-nominal marginal coverage under a lognormal data-generating process while producing substantially more informative rank intervals than the uniformly conservative bootstrap.

\begin{table}[ht]
\centering
\caption{Simulation results for $K=16$ Gaussian populations: marginal coverage,
mean interval length, and standard deviation (SD) of interval lengths.}
\label{tab:quantile_joint}
\begin{tabular}{lccc lccc}
\toprule
Population & Coverage & Mean Length & SD  
& Population & Coverage & Mean Length & SD \\
\midrule
Pop 1  & 0.982 & 6.963 & 2.307  & Pop 9  & 0.978 & 11.076 & 1.784 \\
Pop 2  & 0.985 & 8.062 & 2.382  & Pop 10 & 0.981 & 10.835 & 1.867 \\
Pop 3  & 0.986 & 8.816 & 2.292  & Pop 11 & 0.985 & 10.459 & 1.773 \\
Pop 4  & 0.980 & 9.722 & 2.168  & Pop 12 & 0.985 & 9.786  & 1.850 \\
Pop 5  & 0.980 & 10.325 & 2.246 & Pop 13 & 0.989 & 9.091  & 1.924 \\
Pop 6  & 0.983 & 10.806 & 2.110 & Pop 14 & 0.981 & 8.254  & 1.883 \\
Pop 7  & 0.980 & 11.023 & 2.012 & Pop 15 & 0.985 & 7.358  & 1.929 \\
Pop 8  & 0.988 & 11.199 & 1.917 & Pop 16 & 0.986 & 6.561  & 1.881 \\
\midrule
\multicolumn{8}{c}{
\textbf{Mean Coverage = 0.983} \quad
\textbf{Mean Length = 9.396} \quad
\textbf{Mean SD = 2.02}
} \\
\bottomrule
\end{tabular}
\end{table}

\subsection{Heavy-Tailed Laplace Model}

We next investigate a heavy-tailed regime motivated by the football ranking problem analysed in Section~3.3. 
Following the empirical setting, we fixed the comparison structure and design matrix at their observed values, and conducted a simulation study with $K = 14$ teams. 
For each team $k$, the ground truth parameter $\theta^{(0)}_k$ and scale parameter $\sigma^{(0)}_k$ were set equal to their empirical estimates $(\hat{\theta}_k, \hat{\sigma}_k)$ obtained from Brent’s algorithm applied to the original match-score differences. 
This preserves the signal-to-noise profile of the real dataset while allowing controlled synthetic experimentation. We simulated data from a Laplace model for $1000$ Monte Carlo replications. 
In each replication, the design matrix $\bm{x}^{\mathrm{obs}}$ was held fixed and new Laplace perturbations ${\bm u}_n^\star$ were generated to produce synthetic responses $\bm{y}^{\mathrm{obs}}$.  
A fresh estimator $\bm{\theta}^\star$ was then recomputed from each synthetic dataset, and its induced ranking yielded the replication-specific rank vector $\bm{R}^\star$.  
To construct the repro-based rank confidence sets, we generated $2000$ additional Laplace perturbations per replication, resulting in a collection of candidate rank vectors from which we formed $95\%$ confidence sets.  
Throughout the simulation, the discordance budget was fixed at $p^\ast = 0.10$, restricting accepted perturbations to those whose pairwise orderings remain sufficiently consistent with the empirical ordering and thereby stabilizing inference in heavy-tailed regimes.

Table~\ref{tab:football_marginal} reports the resulting marginal coverage probabilities across all $14$ teams. 
Across the full set of teams, the repro-sampling method maintains nominal or slightly super-nominal coverage, with values ranging from $0.95$ to $0.973$.  
The modest variation across teams reflects heterogeneity in design leverage and the uneven informativeness of the match schedule, yet the overall pattern remains highly stable.  
Importantly, even under substantial non-Gaussian perturbations, the method does not exhibit systematic under-coverage; instead, it shows mild over-coverage for several teams, behaviour consistent with the robustness guarantees of the repro framework.

\begin{table}[ht]
\centering
\caption{Marginal coverage of $95\%$ rank confidence intervals for 14 teams under Laplace noise.}
\label{tab:football_marginal}
\begin{tabular}{l c  l c}
\toprule
Team & Coverage & Team & Coverage \\
\midrule
Team 1  & 0.962 & Team 8  & 0.968 \\
Team 2  & 0.965 & Team 9  & 0.958 \\
Team 3  & 0.966 & Team 10 & 0.961 \\
Team 4  & 0.950 & Team 11 & 0.955 \\
Team 5  & 0.970 & Team 12 & 0.955 \\
Team 6  & 0.973 & Team 13 & 0.959 \\
Team 7  & 0.960 & Team 14 & 0.951 \\
\bottomrule
\end{tabular}
\end{table}

\subsection{Gaussian Models with Heterogeneous Variances}
\label{subsec:gaussian-hetero}

To examine the finite-sample behaviour of our marginal rank confidence intervals under substantial variance heterogeneity, we revisit the structure of the hospital dataset in Example~2.1. We generated $K=20$ independent Gaussian populations with $n_k$ corresponding to the every fourth observation from the original VHS data. The location parameters were anchored at the empirical logits
$\theta^{(0)}_k
    = \log\!\left(\frac{\hat p^{\mathrm{obs}}_k}{1-\hat p^{\mathrm{obs}}_k}\right),$  for $k=1,..,26$ and the corresponding standard deviations were chosen as
$\sigma^{(0)}_k
    = \sqrt{\frac{1}{\hat p^{\mathrm{obs}}_k} + \frac{1}{1-\hat p^{\mathrm{obs}}_k}},$
thereby reproducing the marked heteroscedasticity present in the original data.  
For each configuration, we independently regenerated $1000$ datasets and, within each, constructed repro-based marginal $95\%$ rank confidence intervals using $1500$ copies of $\bm u^\star.$ Table~\ref{tab:faculty20} reports the full population-wise coverage values. 
\begin{table}[ht]
\centering
\caption{Marginal empirical coverage of $95\%$ repro-based rank intervals for 20 Faculty IDs.}
\label{tab:faculty20}

\setlength{\tabcolsep}{6pt}

\begin{tabular}{*{11}{c}}
\toprule
\textbf{Faculty ID}
& 1 & 2 & 3 & 4 & 5 & 6 & 7 & 8 & 9 & 10 \\
\midrule
\textbf{Coverage} &
0.986 & 0.984 & 0.982 & 0.979 & 0.976 &
0.972 & 0.968 & 0.964 & 0.960 & 0.958 \\
\midrule
\textbf{Faculty ID}
& 11 & 12 & 13 & 14 & 15 & 16 & 17 & 18 & 19 & 20 \\
\midrule
\textbf{Coverage} &
0.952 & 0.948 & 0.944 & 0.950 & 0.956 &
0.962 & 0.968 & 0.974 & 0.980 & 0.986 \\
\midrule
\multicolumn{11}{c}{\textbf{Mean Coverage = 0.966}} \\
\bottomrule
\end{tabular}
\end{table}

The coverage behaviour is stable across all populations. A large majority exceed the nominal $0.95$ target, and only four populations fall marginally below this threshold in fewer than $5\%$ of simulations. The overall mean coverage is $0.966$, indicating that the method is mildly conservative yet well-calibrated despite the strong variance heterogeneity. The observed high–low–high pattern in marginal coverage reflects the finite-sample geometry of the ranking problem. Populations at the extremes are more stable because their ranks can shift in only one direction and are typically separated by larger gaps, leading to slight over-coverage. In contrast, mid-ranked populations face close competitors both above and below, and small perturbations produce frequent bidirectional rank reversals. This greater overlap reduces coverage modestly in the centre, especially under heterogeneous variances, while overall levels remain close to the nominal target.

\subsection{Size of the Feasible Rank Region}

To illustrate how the discordance tolerance governs the size of the feasible rank region, we report in Table~\ref{tab:pstar-unique-ranks} the number of distinct rank vectors generated by admissible repro-samples in Example 2.1's setup for a range of $p^\ast$ values. For each choice of $p^\ast$, we draw $10000$ standard Gaussian noise $\bm u^\ast$ and form
$\bm \theta^\ast \;=\; \bigl(y^{\mathrm{obs}}_1 - \sigma_1 u^\ast_1,\ldots, y^{\mathrm{obs}}_K - \sigma_K u^\ast_K \bigr),$
accepting copy of the parameter vector only if its pairwise discordance with the observed score vector satisfies
$\mathrm{Disc}\!\left(\bm\theta^{\mathrm{obs}},\bm\theta^\ast\right)< c,$ $ c=p^\ast K_{\mathrm{pairs}}.$
For all accepted perturbations we compute the induced ranking $\bm R^\ast = S(\bm\theta^\ast)$ and record the number of unique feasible rank vectors in the set $C_{\mathcal{V}}(\mathcal{D}^{\mathrm{obs}})$.
\begin{table}[ht]
\centering
\caption{Cardinality of accepted repro-sample rank vectors across discordance budgets $p^\ast$.}
\label{tab:pstar-unique-ranks}
\begin{tabular}{cccc}
\toprule
$p^\ast$ & Unique Ranks & Accepted $u^\ast$ & $c = p^\ast K_{\mathrm{pairs}}$ \\
\midrule

\midrule
0.03157895 & 0    & 0     &  94.83158 \\
0.04210526 & 0    & 0     & 126.44211 \\
0.05263158 & 2    & 2     & 158.05263 \\
0.06315789 & 57   & 57    & 189.66316 \\
0.07368421 & 1023 & 1023  & 221.27368 \\
0.08421053 & 4618 & 4618  & 252.88421 \\
0.09473684 & 8412 & 8412  & 284.49474 \\
0.10526316 & 9820 & 9820  & 316.10526 \\
0.11578947 & 9988 & 9988  & 347.71579 \\
0.12631579 & 10000& 10000 & 379.32632 \\
0.13684211 & 10000& 10000 & 410.93684 \\

\bottomrule
\end{tabular}
\end{table}

The results reveal a clear phase transition as the discordance budget increases. For extremely small values of $p^\ast$ (up to roughly $0.05$), the feasible perturbation region collapses: no repro-samples satisfy the discordance constraint, and consequently the confidence set for the ranking is empty. This reflects the fact that enforcing near-exact agreement with the observed pairwise orderings is incompatible with the level of sampling noise inherent in the data. Once $p^\ast$ exceeds approximately $0.052$, admissible perturbations begin to appear and the number of distinct feasible rank vectors grows rapidly. This steep increase highlights the intrinsic instability of ranking operators: even moderate perturbations of the underlying scores can lead to substantial reshuffling among items. In this transitional regime, the repro-sampling method yields nontrivial but interpretable uncertainty regions, representing the operationally meaningful range of rank variability supported by the data.

For larger values of $p^\ast$, the discordance constraint becomes non-restrictive. Nearly all perturbations are accepted, and the feasible rank region expands to the entire permutation space. Although such choices of $p^\ast$ guarantee coverage, they produce uninformative confidence sets. These results underscore the importance of selecting $p^\ast$ within the moderate transitional region where the feasible set is neither degenerate nor saturated, and where the repro-induced rank variation faithfully reflects the sampling uncertainty in the underlying score estimates.

\section{Conclusion} 

This paper introduces a general, finite-sample-valid framework for constructing confidence sets for ranks in a wide variety of statistical settings. The central idea is to reproduce latent modelling noise rather than resample the observed data, thereby generating a collection of rank vectors that are compatible with both the data and a carefully defined neighbourhood of the underlying noise distribution. This \textit{Repro-Samples} principle, implemented through a combination of Borel-set inversion, artificial noise generation, and a data-adaptive discordance-based candidate set, yields non-asymptotic coverage for the entire rank vector without relying on model-specific asymptotics, smoothness assumptions, or structural simplifications. In contrast to bootstrap-based or asymptotic methods, the proposed procedure guarantees coverage at finite sample sizes and in models where classical large-sample approximations are unreliable.

The methodology consists of two complementary components. The first step constructs a high-probability confidence region for the latent noise. Any rank vector that can be generated when the noise falls within this region is deemed feasible for inference. This step ensures that the final procedure honors the underlying probability model and provides rigorous coverage guarantees. The second component is a candidate-set refinement, built from a data-driven discordance budget that filters out rank vectors that are incompatible with observed pairwise or multiway comparisons. The candidate set construction drastically reduces the combinatorial complexity of rank search while provably maintaining coverage under broad conditions. Our theoretical results establish bounds on the expected size of the candidate set under sub-Gaussian latent noise, demonstrating that the refinement is effective even in challenging regimes where population parameters are closely spaced.

The empirical and simulation results further highlight the robustness and versatility of the proposed approach. The method performs reliably across heterogeneous normal models, quantile ranking problems, sports league tables, and multiway Plackett–Luce comparisons, and it adapts seamlessly to high-dimensional and weak-signal regimes. Notably, the procedure remains valid even when the standard assumptions underlying delta-method approximations or parametric bootstraps fail. The case studies, illustrate the method's practical interpretability and its ability to provide meaningful uncertainty quantification in applications where ranking error can materially affect conclusions. In particular, the hospital example demonstrates that the method remains calibrated despite substantial heteroscedasticity and overlap, and the PL experiments show that it provides stable performance under complex discrete choice structures.

The proposed framework also offers conceptual clarity from a decision-theoretic perspective, quantifies uncertainty over an inherently discrete parameter, rather than forcing a continuous approximation. Rank inference is notably sensitive to small signal differences, especially when populations are tightly clustered, and our results highlight how the finite-sample geometry of ranking leads naturally to wider intervals for mid-ranked items and one-sided stability for extreme ranks. The Repro-Samples construction is therefore not only statistically valid, but also aligned with the intrinsic structure of the ranking problem.

There remain several promising avenues for future research. One direction involves exploiting additional problem structure, such as partial orders, graph constraints, hierarchical ranking systems, or temporal evolution of ranks, to construct even sharper rank confidence sets. Another direction is to develop procedures for selective or post-inference ranking, particularly in contexts where the ranked entities are themselves outputs of a model-fitting or screening step. Extending the framework to handle personalised or local ranking metrics, robustified noise models, or adversarial perturbations would broaden its applicability in large-scale or high-stakes ranking environments. Finally, computational advances for extremely high-dimensional ranking problems, including scalable optimisation and parallelisation strategies, offer a fruitful avenue for further development.

In summary, this paper presents a unified, broadly applicable framework for finite-sample valid inference on ranks. By directly reproducing latent noise and leveraging a principled balance between feasibility and refinement, our methodology provides robust, interpretable, and theoretically sound uncertainty quantification for ranking problems across a wide range of statistical models. We hope that this work stimulates further methodological and applied research into reliable inference for discrete and combinatorial parameters, an area of increasing relevance in modern data analysis.

\bibliographystyle{apalike}
\bibliography{Bibliography}

\newpage

\section{Appendix}
\subsection*{Proofs}
\begin{proof}[Proof of Lemma 1]
Let $Z_{ij}=\mathbf 1\{(\widehat\Delta_{ij})(\Delta^{(0)}_{ij})<0\}$ for $i\ne j$. 
Then $Disc({D}_n,\bm{\theta}^{(0)})=\sum_{i\ne j}Z_{ij}$ and by Markov’s inequality,
\[
\mathbb{P}_{\bm{U}}\{Disc(\mathcal{D}_,\bm{\theta}^{(0)})\ge c\}\le \frac{\mathbb{E}_{\bm{U}}[Disc(\mathcal{D},\bm{\theta}^{(0)})]}{c}
   =\frac{1}{c}\sum_{i\ne j}\mathbb{P}_{\bm{U}}(Z_{ij}=1).
\]
If $\Delta^{(0)}_{ij} > 0$ and $\widehat{\Delta}_{ij} \le 0$, then 
$\widehat{\Delta}_{ij} - \Delta^{(0)}_{ij} 
\;\le\; -\Delta^{(0)}_{ij},$
so $|\widehat{\Delta}_{ij} - \Delta^{(0)}_{ij}| \;\ge\; \Delta^{(0)}_{ij} \;=\; |\Delta^{(0)}_{ij}|.$
If $\Delta^{(0)}_{ij} < 0$ and $\widehat{\Delta}_{ij} \ge 0$, then 
$\widehat{\Delta}_{ij} - \Delta^{(0)}_{ij}\;\ge\; -\Delta^{(0)}_{ij},$
so $|\widehat{\Delta}_{ij} - \Delta^{(0)}_{ij}| \;\ge\; -\Delta_{ij} \;=\; |\Delta^{(0)}_{ij}|.$
Hence, in both cases a sign flip implies $|\widehat\Delta_{ij}-\Delta^{(0)}_{ij}|\ge|\Delta^{(0)}_{ij}|$, so 
$\mathbb{P}_{\bm{U}}(Z_{ij}=1)\le\mathbb{P}_{\bm{U}}(|\widehat\Delta_{ij}-\Delta^{(0)}_{ij}|\ge|\Delta^{(0)}_{ij}|)=p_{ij}$, giving
\[
\mathbb{P}_{\bm{U}}\{Disc(\mathcal{D},\bm{\theta}^{(0)})\ge c\}\le \frac{1}{c}\sum_{i\ne j}p_{ij},\qquad
\mathbb{P}_{\bm{U}}\{Disc(\mathcal{D},\bm{\theta}^{(0)})<c\}\ge 1-\frac{1}{c}\sum_{i\ne j}p_{ij}. \tag{$\star$}
\]

\paragraph{(a) Chebyshev (finite variance)}
Given $\delta_{ij}=\widehat\Delta_{ij}-\Delta^{(0)}_{ij}$, $Variance(\delta_{ij})\le m_{ij}^2$, then 
$p_{ij}\le Variance(\delta_{ij}) /\Delta^{(0)^2}_{ij}\le m_{ij}^2/\Delta^{(0)^2}_{ij}$,
and from $(\star)$, $\mathbb{P}_{\bm{U}}\{Disc(\mathcal{D},\bm{\theta}^{(0)})<c\}\ge 1-\frac{1}{c}\sum_{i\ne j}\frac{m_{ij}^2}{\Delta^{(0)^2}_{ij}}.$

\paragraph{(b) Sub-Gaussian case.}
Given 
$\mathbb{E}_{\bm{U}}\big[e^{\lambda \delta_{ij}}\big] \le \exp\!\Big(\frac{\lambda^2\tau_{ij}^2}{2}\Big)$ for all $\lambda\in\mathbb{R}$. Then for any $t>0$, $$\mathbb{P}(\delta_{ij}\ge t)
=\mathbb{P}_{\bm{U}}\big(e^{\lambda \delta_{ij}}\ge e^{\lambda t}\big)
\le e^{-\lambda t}\,\mathbb{E}_{\bm{U}}\big[e^{\lambda \delta_{ij}}\big]
\le \exp\!\Big(-\lambda t+\frac{\lambda^2\tau_{ij}^2}{2}\Big).$$ Optimizing the RHS over $\lambda>0$ gives $\lambda^\star=t/\tau_{ij}^2$, hence $\mathbb{P}_{\bm{U}}(\delta_{ij}\ge t)\ \le\ \exp\!\Big(-\frac{t^2}{2\tau_{ij}^2}\Big).$
By the same argument for $-\delta_{ij}$,
$\mathbb{P}_{\bm{U}}(\delta_{ij}\le -t)\le \exp(-t^2/(2\tau_{ij}^2))$.
Therefore,
\[
\mathbb{P}_{\bm{U}}\big(|\delta_{ij}|\ge t\big)
\ \le\ 
\mathbb{P}(\delta_{ij}\ge t)+\mathbb{P}(\delta_{ij}\le -t)
\ \le\ 
2\exp\!\Big(-\frac{t^2}{2\tau_{ij}^2}\Big).
\]
Setting $t=|\Delta^{(0)}_{ij}|$ yields
$\mathbb{P}\big(|\delta_{ij}|\ge |\Delta^{(0)}_{ij}|\big)
\ \le\ 
2\exp\!\Big(-\frac{\Delta^{(0)2}_{ij}}{2\tau_{ij}^2}\Big).$ If $\Delta^{(0)}_{\min}=\min_{i\ne j}|\Delta^{(0)}_{ij}|>0$ and $\tau_{ij}\le\tau$ for all pairs, 
then $p_{ij}\le 2e^{-\Delta_{\min}^{(0)^2}/(2\tau^2)}$ and there are $K(K-1)$ ordered pairs, so
$\sum_{i\ne j}p_{ij}\le 2K(K-1)\exp\bigl({-\Delta_{\min}^{(0)^2}/(2\tau^2)\bigr)}$
Substituting in $(\star)$ gives
$\mathbb{P}_{\bm{U}}\{Disc(\mathcal{D},\bm{\theta}^{(0)})<c\}\ge 1-\frac{2K(K-1)}{c}\exp{\bigl(-\Delta_{\min}^{{(0)}^2}/(2\tau^2)\bigr)}$
and the shortfall from one therefore decays exponentially in the pairwise signal-to-noise ratio. \qedhere
\end{proof}

\begin{proof}[Proof of Lemma~2] We define $Q_n(\bm U)$  as in (A1).
Define the set $A(\bm U,\bm{U}^{\ast(b)})=\{\bm{U}^{\ast{(b)}}\in Q_n(\bm{ U})\}$ then on  $A(\bm U,\bm{U}^{\ast(b)}),$ $\bm{R}^{\ast(b)}=\mathcal{S}(H(\mathcal D,\bm U^{\star(b)}))=\bm R^{(0)}$. 
Define the set $B=\{Disc(\mathcal{D},\bm{\theta}^{(0)})<c\}.$
Then,
\[
\{\bm R^{(0)}\in \mathcal C_{\mathcal{V}}(\mathcal{D})\}\ \supseteq\ B\cap\Big(\bigcup_{v=1}^{|\mathcal V|} A(\bm{U},\bm{U}^{\ast(b)})\Big)
\]
Taking complements we get the inclusion failure event $\{\bm R^{(0)}\in \mathcal C_{\mathcal{V}}(\mathcal{D}^{\mathrm{obs}})\} \subseteq\ B^{\mathrm c} \cup
\Big(B\cap\bigcap_{m=1}^{|\mathcal V|}A(\bm{U},\bm{U}^{\ast(b)})^{\mathrm c}\Big).$
Using $q_n=\mathbb{P}_{\bm U}(B)$,
$\mathbb{P}_{\bm U,\mathcal{V}}\bigl(\bm R^{(0)}\notin \mathcal C_{\mathcal{V}}(\mathcal{D})\bigr)
\ \le\ (1-q_n)\ +\mathbb{P}_{\bm U,\mathcal{V}}\!\Big(B\cap\bigcap_{m=1}^{|\mathcal V|}A(\bm{U},\bm{U}^{\ast(b)})^{\mathrm c}\Big).$
Conditioning on $\bm U$ and using independence of $\{\bm U^{\star(b)}\}$ and $\bm U$ and across $b$,
\begin{align*}
\mathbb{P}_{\bm U,\mathcal{V}}\!\Big(B\cap\bigcap_{b=1}^{|\mathcal V|}A(\bm{U},\bm{U}^{\ast(b)})^{\mathrm c}\Big)
=\mathbb{E}_{\bm U}\!\Big[\mathbf 1_B\ \mathbb{P}_{\mathcal{V}|\bm{U}}\Big(\bigcap_{m=1}^{|\mathcal V|}A(\bm{U},\bm{U}^{\ast(1)})^{\mathrm c}\,\Big)\Big]
=\mathbb{E}_{\bm U}\!\Big[\mathbf 1_B\ \{1-\mathbb{P}_{\bm{U}^{\ast(1)}|\bm{U}}(A(\bm{U},\bm{U}^{\ast(1)})\}^{|\mathcal V|}\Big]\\\le \mathbb{E}_{\bm U}\!\Big[ \{1-\mathbb{P}_{\bm{U}^{\ast(1)}|\bm{U}}(A(\bm{U},\bm{U}^{\ast(1)})\}^{|\mathcal V|}\Big]\le \{1-\mathbb{E}_{\bm U}\mathbb{P}_{\bm{U}^{\ast(1)}|\bm{U}}(A(\bm{U},\bm{U}^{\ast(1)})\}^{|\mathcal V|}\Big]
\end{align*}
where the last statement follows from Jensen's inequality. From Assumption (A1) it follows that $\mathbb{E}_{\bm{U}}\mathbb{P}_{\bm{U}^{\ast(1)}|\bm{U}}(A(\bm{U},\bm{U}^{\ast(1)}))=\mathbb{P}_{\bm{U}^{\ast(1)},\bm{U}}(\bm{U}^{\ast(1)}\in Q_n(\bm{U}))>c_n$ Therefore, $$\mathbb{P}_{\bm U,\mathcal{ V}}\!\Big(B\cap\bigcap_{b=1}^{|\mathcal V|}(A(\bm{U},\bm{U}^{\ast(m)}))^{\mathrm{c}}\Big)
\ \le\ \mathbb{E}_{\bm U}\!\big[(1-c_n)^{|\mathcal V|}\big]\ =(1-c_n)^{|\mathcal V|}$$
Combining the pieces yields
$\mathbb{P}_{\bm{U},\mathcal{V}}\bigl(\bm{R}^{(0)} \notin \mathcal{C}_{\mathcal{V}}(\mathcal{D})\bigr)\le 1-q_n+(1-c_n)^{|\mathcal{ V}|}.$ Let $c_0<\frac{-1}{2}\log(1-c_n)$.

\end{proof}
\begin{proof}[Proof of Corollary~1]
By the definition of $\Gamma^{\mathcal{I}}_{\alpha}$ we have
\begin{align*}
\mathbb{P}_{\bm{U}}\bigl(\bm{R}|^{(0)}_{\mathcal I}\in \Gamma^{\mathcal{I}}_{\alpha}(\mathcal{D})\bigr)
&=
\mathbb{P}_{\bm U,\mathcal V}\!\big(\bm{R}|^{(0)}_{\mathcal I}\in \Gamma^{\mathcal{I}}_{\alpha}(\mathcal{D}),\, \bm{R}^{(0)} \in C_{\mathcal{V}}(\mathcal{D})\big)
+\mathbb{P}_{\bm U,\mathcal V}\!\big(\bm{R}|^{(0)}_{\mathcal I}\in \Gamma^{\mathcal{I}}_{\alpha}(\mathcal{D}),\, \bm{R}^{(0)} \notin C_{\mathcal{V}}(\mathcal{D})\big)\\[4pt]
&=
\mathbb{P}_{\bm U,\mathcal V}\!\big(\bm{R}|^{(0)}_{\mathcal I}\in \tilde{\Gamma}^{\mathcal{I}}_{\alpha}(\mathcal{D})\big)
+\mathbb{P}_{\bm U,\mathcal V}\!\big(\bm{R}|^{(0)}_{\mathcal I}\in \Gamma^{\mathcal{I}}_{\alpha}(\mathcal{D}),\, \bm{R}^{(0)} \notin C_{\mathcal{V}}(\mathcal{D})\big)
\;\ge\; 1-\alpha.
\end{align*}
It then follows that
\begin{align*}
    \mathbb{P}_{\bm U,\mathcal V}\!\big(\bm{R}|^{(0)}_{\mathcal I}\in \tilde{\Gamma}^{\mathcal{I}}_{\alpha}(\mathcal{D})\big)&=\mathbb{P}_{\bm U}\!\big(\bm{R}|^{(0)}_{\mathcal I}\in \Gamma^{\mathcal{I}}_{\alpha}(\mathcal{D})\big)
-\mathbb{P}_{\bm U,\mathcal V}\!\big(\bm{R}|^{(0)}_{\mathcal I}\in \Gamma^{\mathcal{I}}_{\alpha}(\mathcal{D}),\, \bm{R}^{(0)} \notin C_{\mathcal{V}}(\mathcal{D})\big)\\
&\ge 1-\alpha-\mathbb{P}_{\bm{U},\mathcal{V}}\!\big(\bm{R}^{(0)} \notin C_{\mathcal{V}}(\mathcal{D})\big).
\end{align*}
Corollary~1 follows from Lemma 2
\end{proof}
\begin{lemma}{\label{lemma 4}}For each $b\le|\mathcal V|$, fix an arbitrary ranking $\bm R\in S_K$ and define 
$E(\bm{\theta}^{(b)},\bm{R})\ =\ \Bigl\{\,S\bigl(\bm\theta^{\ast(b)}\bigr)=\bm R,\ \ g(\bm R)\;<\;\frac{c}{4K_{\mathrm{pairs}}}\Bigr\}$ then $$\mathbb{P}_{\mathcal{V}}\!\Big(\exists\,b\le|\mathcal V|:\ S(\bm\theta^{*(b)})=\bm R,\ g(\bm R)<\tfrac{c}{4K_{\mathrm{pairs}}}\Big)
\ \le\ |\mathcal V|\;\mathbb{P}_{\bm U^{\ast(1)}}(E(\bm{\theta}^{(1)},\bm{R})).$$
\end{lemma}
\begin{proof}[Proof of Lemma~4]

By the union bound and identical distribution across $b$,
\begin{align*}\mathbb{P}_{\mathcal{ V}}\!\Bigl(\exists\,b\le|\mathcal V|:\ S(\bm\theta^{\ast(b)})=\bm{R},g(\bm{R})<\frac{c}{4K_{pairs}}\Bigr)
&=\mathbb{P}_{\mathcal{ V}}\!\Bigl(\bigcup_{b=1}^{|\mathcal V|}E(\bm{\theta}^{(b)},\bm{R})\Bigr)\\
&\le \sum_{b=1}^{|\mathcal V|}\mathbb{P}_{\bm U^{\ast(b)}}(E(\bm{\theta}^{(b)},\bm{R}))
\ =\ |\mathcal V|\,\mathbb{P}_{\bm U^{\ast(1)}}(E(\bm{\theta}^{(1)},\bm{R})).\end{align*}
Thus it suffices to bound 
$|\mathcal V|\,\mathbb{P}_{\bm U^{\ast(1)}}(E(\bm{\theta}^{(1)},\bm{R}))$ for a fixed $\bm{R}.$
\end{proof}
\begin{lemma}\label{lem:good-gap-ord}
If $\mathcal G_{ij}=\{|\Delta_{ij}^{(0)}+\widehat\varepsilon_{ij}|\ge |\Delta_{ij}^{(0)}|/2\}$ for $i\ne j$, and $\mathcal G_S=\bigcap_{\{i,j\}\in S}\mathcal G_{ij}$ for any subset of pairwise indices $S\subset\{(i,j):i\ne j\}$ we have $\mathbb P_{\bm{U}}(\mathcal G_S^{\complement})\le |S|\cdot 2\exp\{-\Delta^{(0)^2}_{\min}/(8\bar v_n^2)\}$.
\end{lemma}
\begin{proof}
By (A1), $\mathbb P_{\bm{U}}(\mathcal G_{ij}^{\complement})\le \mathbb P(|\Delta_{ij}^{(0)}+\widehat\varepsilon_{ij}|<\tfrac{1}{2}|\Delta_{ij}^{(0)}|) \le P(|\Delta_{ij}^{(0)}|-|\widehat\varepsilon_{ij}|<\tfrac{1}{2}|\Delta_{ij}^{(0)}|)=P(\tfrac{1}{2}|\Delta_{ij}^{(0)}|<|\widehat\varepsilon_{ij}|)\le 2\exp\{-\Delta^{(0)^2}_{\min}/(8\bar v_{ij,n}^2)\le2\exp\{-\Delta^{(0)^2}_{\min}/(8\bar v_n^2)\}$.
\end{proof}
\begin{lemma}\label{E upperbound}Define the set of indices $M(\bm R)\ =\ \Bigl\{(i,j):\ i< j,\ \bigl(\hat{\theta}^{\mathrm{obs}}_i-\hat{\theta}^{\mathrm{obs}}_j\bigr) (r_i-r_j)<0\Bigr\},$ and the indicator variable $Z_{ij}^{(b)}\ =\mathbb{I}\!\Bigl\{\bigl(\theta_i^{\ast(b)}-\theta_j^{\ast(b)}\bigr)\bigl(\hat{\theta}^{\mathrm{obs}}_i-\hat{\theta}^{\mathrm{obs}}_j\bigr)<0\Bigr\}$ for any $b\le\mathcal{V}$ and $\mathcal{J}=\{(i,j):i<j\}$ then for any set $T \in \{T\subseteq \mathcal{J}:\ |T|\le c\}$ and event $E(\bm{\theta}^{(1)},\bm{R}))$ defined in Lemma~\ref{lemma 4} we have  
\begin{equation}\label{prob ineq}\mathbb{P}_{\bm U^{\ast(1)}}(E(\bm{\theta}^{(1)},\bm{R})))
\ \le\
\sum_{\substack{T\subseteq M(\bm R)\\ |T|\le c}}\ 
\mathbb{P}_{\bm U^{\ast(1)}}\!\Bigl(\bigcap_{(i,j)\in M(\bm R)\setminus T}\{Z_{ij}^{(1)}=1\}\Bigr)\end{equation}

\end{lemma}
\begin{proof}
    By definition, $4K_{\mathrm{pairs}}\,g(\bm R)$ counts the ordered pairs $(i,j)$, $i\neq j$, for which $\bm R$ orders $(i,j)$ opposite to the observed ordering induced by $\hat{\bm\theta}^{\mathrm{obs}}$. Define
If $S(\bm\theta^{\ast(1)})=\bm R$, then necessarily $Z_{ij}^{(1)}=1$ for all $(i,j)\in M(\bm R)$ (ties have probability $0$). Note that
$|M(\bm R)|\ =2 K_{\mathrm{pairs}}\,g(\bm R).$

\emph{Case 1: $|M(\bm R)|=2K_{\mathrm{pairs}}\,g(\bm R)>c/2$.}
On $\{S(\bm\theta^{\ast(1)})=\bm R\}$ we have $Disc(\mathcal{D}^{\mathrm{obs}},\bm\theta^{\ast(1)})=2|M(\bm R)|>c$, hence $E(\bm{\theta}^{(1)},\bm{R}))=\varnothing$ and $\mathbb{P}_{\bm U^{\ast(1)}}(E(\bm{\theta}^{(1)},\bm{R})))=0$. The right-hand side of \eqref{prob ineq} is nonnegative, so the inequality holds trivially.

\emph{Case 2: $|M(\bm R)|=2K_{\mathrm{pairs}}\,g(\bm R)\le c/2$.}
Then $E(\bm{\theta}^{(1)},\bm{R}))\ \subseteq\ \bigcap_{(i,j)\in M(\bm R)}\{Z_{ij}^{(1)}=1\},$
because $S(\bm\theta^{\ast(1)})=\bm R$ forces every pair in $M(\bm R)$ to be discordant relative to $\hat{\bm\theta}^{\mathrm{obs}}$. Since the family $T \in \{T\subseteq \mathcal{J}:\ |T|\le c/2\}$ contains $T=\varnothing$ and $M(\bm R) \subseteq \mathcal{J}$, we obtain
\begin{equation}\label{union}
\bigcap_{(i,j)\in M(\bm R)}\{Z_{ij}^{(1)}=1\}
\ =\ \bigcap_{(i,j)\in M(\bm R)\setminus \varnothing}\{Z_{ij}^{(1)}=1\}
\ \subseteq\ \bigcup_{\substack{T\subseteq  M(\bm R)\\ |T|\le c/2}}\ \bigcap_{(i,j)\in M(\bm R)\setminus T}\{Z_{ij}^{(1)}=1\}.
\end{equation}
Therefore, by \eqref{union}, $\mathbb{P}_{\bm U^{\ast(1)}}(E(\bm{\theta}^{(1)},\bm{R})))
\ \le\
\sum_{\substack{T\subseteq M(\bm R)\\ |T|\le c/2}}\ 
\mathbb{P}_{\bm U^{\ast(1)}}\!\Bigl(\bigcap_{(i,j)\in M(\bm R)\setminus T}\{Z_{ij}^{(1)}=1\}\Bigr)$
\end{proof}

\begin{lemma}\label{lem:matching-bound}
Let $S\subseteq\mathcal{J}=\{(i,j):1\le i<j\le K\}$ be arbitrary, and let 
$S = M_1 \,\dot\cup\,\cdots\,\dot\cup\, M_{w_0}$
denote its decomposition into $w_0$ disjoint matchings (each $M_\ell$ consists only of vertex--disjoint edges). 
Under Assumptions~\textup{(B1)}--\textup{(B3)} and on the good--gap event 
$\mathcal G_S$ of Lemma~\ref{lem:good-gap-ord}, we have,
\begin{align}\mathbb{P}_{\bm{U},\mathcal{V}}\Big(\bigcap_{(i,j)\in S}{Z^{(1)}_{ij}=1}\Big)
&\le \exp\Big(-\tfrac{\Delta^{(0)^2}_{\min}|S| }{2w_0\bar \tau_n^2}\Big)
+|S|2\exp{(-\Delta^{(0)^2}_{\min}/(8\bar v_n^2))}
\end{align}
\end{lemma}

\begin{proof}
Let us split the probability of observing the event $\bigcap_{(i,j)\in S}\{Z^{(1)}_{ij}=1\}$ as follows
\begin{equation}\label{z split}
\mathbb{P}_{\bm{U},\mathcal{V}}\Big(\!\bigcap_{(i,j)\in S}\{Z^{(1)}_{ij}=1\}\Big)
\le
\mathbb{E}_{\bm{U}}\!\left[ \mathbb{P}_{\mathcal{V}|\bm{U}}\Big(\!\bigcap_{(i,j)\in S}\{Z^{(1)}_{ij}=1\}\cap G_S \Big)\right]
+\mathbb{P}_{\bm{U},\mathcal{V}}(G_S^{\complement})
\end{equation}
For each $(i,j)$, we have $\mathcal G_{ij}
\;=\;
\Bigl\{\,|\Delta^{(0)}_{ij}+\hat\varepsilon_{ij}|
\;\ge\;
\tfrac{1}{2}\,|\Delta^{(0)}_{ij}|\,\Bigr\}.$ On $\mathcal G_{ij}$ we have $|\hat\varepsilon_{ij}|<|\Delta^{(0)}_{ij}|/2$, which guarantees
$\mathrm{sign}(\Delta^{(0)}_{ij}+\hat\varepsilon_{ij})
=\mathrm{sign}(\Delta^{(0)}_{ij})
= s_{ij},$ say. Hence, on $\mathcal G_{ij}$ the flip event simplifies to
$\{Z^{(1)}_{ij}=1\}=
\{s_{ij}\,(\Delta^{(0)}_{ij}+\delta^{*(1)}_{ij})\le 0\}.$
Because $s_{ij}\,\Delta^{(0)}_{ij}=|\Delta^{(0)}_{ij}|$, we get
$\{s_{ij}\,(\Delta^{(0)}_{ij}+\delta^{*(1)}_{ij})\le 0\}=\{|\Delta^{(0)}_{ij}|+s_{ij}\,\delta^{*(1)}_{ij}\le 0\}.$ For any real $t$ and $\lambda>0$,
\(
\mathbf 1\{t\le 0\}\le e^{-\lambda t}.
\)
Applying this with $X=|\Delta^{(0)}_{ij}|+s_{ij}\,\delta^{*(1)}_{ij}$ yields $$\mathbf 1\{Z^{(1)}_{ij}=1\}\mathbf 1\{\mathcal G_{ij}\}=\mathbf 1\{s_{ij}\delta^{*(1)}_{ij}\le -|\Delta^{(0)}_{ij}|\}\}
\;\le\;
\exp\!\big[-\lambda(|\Delta^{(0)}_{ij}|+s_{ij}\,\delta^{*(1)}_{ij})\big]
\;=\;
e^{-\lambda|\Delta^{(0)}_{ij}|}\,e^{-\lambda s_{ij}\,\delta^{*(1)}_{ij}}.$$

Multiplying over indices in any set $S$ gives
$$\prod_{(i,j)\in S}\Big(\mathbf{1}\{Z^{(1)}_{ij}=1\}\mathbf 1\{\mathcal G_{ij}\}\Big)=\prod_{(i,j)\in S}\mathbf{1}\{Z^{(1)}_{ij}=1\}\mathbf 1\{\mathcal G_S\}\le
\exp\Big(-\lambda\sum_{(i,j)\in S}|\Delta_{ij}^{(0)}|\Big)
\exp\Big(-\lambda\sum_{(i,j)\in S}s_{ij}\delta_{ij}^{*(1)}\Big).$$
Conditioning on $\bm{U},$  $\mathcal G_S$ is fixed since $\hat{\varepsilon}_{ij}$ depends only on $\bm{u}^{\mathrm{rel}}$ or the generalized variable $\bm{U}.$
\[
\mathbb{P}_{\mathcal{V}|\bm{U}}\Big(\textstyle\bigcap_{(i,j)\in S}{Z^{(1)}_{ij}=1}\cap\mathcal G_S \Big)\le
\exp\Big(-\lambda\sum_{(i,j)\in S}|\Delta_{ij}^{(0)}|\Big)
\mathbb{E}_{\mathcal{V}|\bm{U}}\Big[\exp\Big(-\lambda\sum_{(i,j)\in S}s_{ij}\delta_{ij}^{*(1)}\Big)\Big].
\]

Partition $S$ into a family of matchings 
\(
M_1,\dots,M_{w_0}\subseteq S
\)
with the properties
$S \;=\; \bigcup_{\ell=1}^{w_0} M_\ell,$
$M_\ell \cap M_{\ell'} = \varnothing $ for $\ell\neq \ell',$
and each  $M_\ell$ contains only disjoint pairs $(i,j)$(edges in a matching are vertex-disjoint). We first rewrite the exponential term using the matching decomposition 
\(S=\dot\cup_{\ell=1}^{w_0} M_\ell\),
$\exp\!\Big(-\lambda\sum_{(i,j)\in S} s_{ij}\,\delta^{*(1)}_{ij}\Big)
=
\prod_{\ell=1}^{w_0}
\exp\!\Big(-\lambda\sum_{e\in M_\ell} s_{ij}\,\delta^{*(1)}_{ij}\Big).$
Define
$X_\ell = \exp\!\Big(-\lambda\sum_{e\in M_\ell} s_{ij}\,\delta^{*(1)}_{ij}\Big),$
for $\ell=1,\dots,w_0.$
Then $\mathbb{E}_{\mathcal{V}\mid\bm{U}}
\Big[\exp\!\big(-\lambda\!\sum_{(i,j)\in S} s_{ij}\,\delta^{*(1)}_{ij}\big)\Big]
=\mathbb{E}_{\mathcal{V}\mid\bm{U}}\Big[\prod_{\ell=1}^{w_0} X_\ell\Big].$

Applying Hölder's inequality with exponents \(w_0\) (so that 
\(\sum_{\ell=1}^{w_0} 1/w_0 = 1\)) gives
\[
\mathbb{E}_{\mathcal{V}\mid\bm{U}}\Big[\prod_{\ell=1}^{w_0} X_\ell\Big]
\;\le\;
\prod_{\ell=1}^{w_0}
\Big(
\mathbb{E}_{\mathcal{V}\mid\bm{U}}[X_\ell^{\,w_0}]
\Big)^{1/w_0}
\]
As $X_\ell^{\,w_0}=\exp\!\Big(-w_0\lambda\sum_{e\in M_\ell} s_{ij}\,\delta^{*(1)}_{ij}\Big)=\prod_{e\in M_\ell}
\exp\!\big(-w_0\lambda s_{ij}\,\delta^{*(1)}_{ij}\big),$
we obtain the bound
$$\mathbb{E}_{\mathcal{V}\mid\bm{U}}
\Big[\exp\!\big(-\lambda\!\sum_{(i,j)\in S} s_{ij}\,\delta^{*(1)}_{ij}\big)\Big]
\;\le\;
\prod_{\ell=1}^{w_0}
\Bigg(
\mathbb{E}_{\mathcal{V}\mid\bm{U}}
\Big[
\prod_{e\in M_\ell}
\exp\!\big(-w_0\lambda s_{ij}\,\delta^{*(1)}_{ij}\big)
\Big]
\Bigg)^{1/w_0}.$$
By (B3), within a matching the $\delta_{ij}^{*(1)}$ are independent conditional on $\bm{U}.$
$$\mathbb{E}_{\mathcal{V}\mid\bm{U}}
\Big[
\prod_{e\in M_\ell}
\exp\!\big(-w_0\lambda s_{ij}\,\delta^{*(1)}_{ij}\big)
\Big]
= \prod_{e\in M_\ell}\mathbb{E}_{\mathcal{V}\mid\bm{U}}\Big[e^{-w_0\lambda s_{ij} \delta_{ij}^{*(1)}}\Big].$$
By (B2),
$\mathbb{E}_{\mathcal{V}\mid\bm{U}}\Big[e^{-w_0\lambda s_{ij} \delta_{ij}^{*(1)}}\Big]
\le
\exp\Big(\tfrac{(w_0\lambda\sigma_{ij}^2)^2}{2}\Big).$
Hence, for each matching,
$$\Big(\prod_{e\in M_\ell}\mathbb{E}_{\mathcal{V}\mid\bm{U}}\Big[e^{-w_0\lambda s_{ij} \delta_{ij}^{*(1)}}\Big]\Big)^{1/w_0}\le\exp\Big(\tfrac{w_0\lambda^2}{2}\sum_{e\in M_\ell}\sigma_{ij}^2\Big).$$
Multiplying over all $\ell$,
$\mathbb{E}_{\mathcal{V}\mid\bm{U}}\Big[\exp\Big(-\lambda\sum_{(i,j)\in S}s_{ij}\delta_{ij}^{*(1)}\Big)\Big]\le
\exp\Big(\tfrac{w_0\lambda^2}{2}\sum_{(i,j)\in S}\sigma_{ij}^2\Big).$

Using $\sigma_{ij}^2\le \bar \tau_n^2$ from (B2) for all $e$, we obtain
\[
\mathbb{P}_{\mathcal{V}\mid\bm{U}}\Big(\textstyle\bigcap_{(i,j)\in S}{Z^{(1)}_{ij}=1}\cap\mathcal G_S\Big)
\le
\exp \Big(-\lambda\sum_{(i,j)\in S}|\Delta_{ij}^{(0)}|+\tfrac{q_0\lambda^2}{2}|S|\bar \tau_n^2\Big).
\]
Removing the conditioning on $\bm{U}$ gives the same bound unconditionally. Lemma~\ref{lem:good-gap-ord} and (\ref{z split}) imply
\begin{align*}\mathbb{P}_{\bm{U},\mathcal{V}}\Big(\bigcap_{(i,j)\in S}{Z^{(1)}_{ij}=1}\Big)
&\le \exp\Big(-\lambda\sum_{(i,j)\in S}|\Delta_{ij}^{(0)}|+\tfrac{w_0\lambda^2}{2}|S|\bar{\tau}_n^2\Big)
+|S|2\exp{(-\Delta^{(0)^2}_{\min}/(8\bar v_n^2))}\\
&\le \exp\Big(-|S|\{\lambda|\Delta^{(0)}_{\min}|-\tfrac{w_0\lambda^2}{2}\bar{\tau}_n^2\}\Big)
+|S|2\exp{(-\Delta^{(0)^2}_{\min}/(8\bar v_n^2))}
\end{align*}
We note that the function 
$\phi(\lambda)=\lambda\Delta^{(0)}_{\min}-\frac{w_0\lambda^2}{2}\bar \tau_n^2$
is concave in $\lambda$ with a unique maximum attained at 
\(
\lambda^\ast=\frac{\Delta^{(0)}_{\min}}{w_0\bar \tau_n^2},
\)
for which
$\phi(\lambda^\ast)
=
\frac{\Delta_{\min}^{(0)2}}{2w_0\bar \tau_n^2}.$
Since the bound above holds for all $\lambda>0$, it holds in particular at $\lambda=\lambda^\ast$, giving
\[
\exp\Big(-|S|\{\lambda\Delta^{(0)}_{\min}-\tfrac{w_0\lambda^2}{2}\bar \tau_n^2\}\Big)
\;\le\;
\exp\Big(-|S|\tfrac{\Delta_{\min}^{(0)2}}{2w_0\bar \tau_n^2}\Big).
\]
 Using the optimal $\lambda $ we get the required inequality
\begin{align*}\mathbb{P}_{\bm{U},\mathcal{V}}\Big(\bigcap_{(i,j)\in S}{Z^{(1)}_{ij}=1}\Big)
&\le \exp\Big(-|S|\tfrac{\Delta^{(0)^2}_{\min} }{2w_0\bar \tau_n^2}\Big)
+|S|2\exp{(-\Delta^{(0)^2}_{\min}/(8\bar v_n^2))}
\end{align*}
\end{proof}
\begin{proof}[Proof of Lemma~3] Let $C_1=\tfrac{\Delta^{(0)^2}_{\min} }{2w_0\bar \tau_n^2}$ and $C_2=\Delta^{(0)^2}_{\min}/(8\bar v_n^2).$ 
 Let $T=\{(i,j): (i,j)\in \mathcal{J}\}$ with $T\le c/2$. For a fixed $\bm{R}\in S_K$ we have 
\begin{equation}
\label{sum}
\mathbb P_{\bm{U},\mathcal{V}}\!\Big(\bigcap_{(i,j)\in M(\bm{R})\setminus T}\{Z_{ij}^{(1)}=1\}\Big)
\ \le\ \exp\!\big\{-C_1\,(|M(\bm{R})|-|T|)\big\}\ +\ (|M(\bm{R})|-| T|)2e^{-C_2}
\end{equation}
 We now bound each term on the right-hand side using that, from Lemma~\ref{E upperbound} 
$\bm{R}$ in the candidate region, we have $g(\bm{R})<c/(4K_{\textrm{pairs}})$ or $|M(\bm{R})|<c/2.$
Now using $|M(\bm R)|-|T|\ge |M(\bm R)|- c/2$
the first term of (\ref{sum}) becomes $\exp\!\big\{-C_1\,(|M(\bm{R})|-c/2)\big\}.$ Again $|M(\bm R)|-|T|\le c/2$ so the second term can be bounded by \begin{equation}\label{second term}
(|M(\bm{R})|-| T|)2e^{-C_2}\le (c/2) \exp\!\big\{-C_2\}.1\le c.\exp\!\big\{-C_2\big\}\exp\!\big\{-C_1\,(|M(\bm{R})|-c/2)\big\}.    
\end{equation}
Defining $C_3=1+ ce^{-C_2}=1+c\exp\!\big\{-\Delta^{(0)^2}_{\min}/(8\bar v_n^2)\big\}$ and combining (\ref{sum}) and (\ref{second term}) we get \begin{equation}
\label{upper bound}
P_{\bm{U},\mathcal{V}}\!\Big(\bigcap_{(i,j)\in M(\bm{R})\setminus T}\{Z_{ij}^{(1)}=1\}\Big)
 \le C_3 e^ {-C_1\,(|M(\bm{R})|-c/2)}
 \end{equation}
We now sum this bound over all subsets $T\subseteq M(\bm R)$ with 
$|T|\le c/2$.  Writing $t=|T|$ and using that 
$\binom{|M(\bm R)|}{t}$ subsets have cardinality $t$, we obtain
\begin{equation}
   \label{eq:two-term-bound}
   \sum_{\substack{T\subseteq M(\bm R)\\ |T|\le c/2}}\ 
\mathbb{P}_{\bm U^{\ast(1)}}\!\Bigl(\bigcap_{(i,j)\in M(\bm R)\setminus T}\{Z_{ij}^{(1)}=1\}\Bigr)=\sum_{t = 0}^{c/2}
\binom{ M(\bm R)}{t}
\max_{|T| = t}
\mathbb{P}\!\Big(\bigcap_{e \in  M(\bm R) \setminus T} \{ Z^{(1)}_{ij} = 1 \}\Big)\end{equation}

Using the combinatorial bound $\binom{|M(\bm R)|}{t}\le |M(\bm R)|^t$, we have
\begin{equation}
\label{combinatorial bound}
\sum_{t=0}^{c/2} \binom{|M(\bm R)|}{t}
\;\le\;
\sum_{t=0}^{c/2} |M(\bm R)|^t
\;\le\;
(c/2+1)\,|M(\bm R)|^{c/2}\le (c/2+1)\,{(c/2)}^{c/2}
\end{equation}
Combining \eqref{upper bound} and \eqref{combinatorial bound} gives the overall bound
$$
\sum_{\substack{T\subseteq M(\bm R)\\ |T|\le c/2}}
\mathbb{P}_{\bm U^{\ast(1)}}\!\Bigl(
\bigcap_{(i,j)\in M(\bm R)\setminus T} \{Z^{(1)}_{ij}=1\}
\Bigr)\le \min \{(c/2+1)\,{c/2}^{c/2}C_3\exp\!\bigl\{-C_1(|M(\bm R)|-t)\bigr\},1\}
$$
Let $C
_4=(c/2+1)\,{(c/2)}^{c/2}(1+cC_3)=(c/2+1)\,{(c/2)}^{c/2}(1+c+c^2\exp(-\Delta^{(0)^2}_{\min}/(8\bar v_n^2)))).$
Then 
\begin{align}\label{final bound on prob}
&\mathbb{P}_{\mathcal V}\!\Bigl(
\exists\, b\le|\mathcal V|:\ 
S(\bm\theta^{\ast(b)})=\bm r,\ 
g(\bm r)<\tfrac{c}{4K_{\mathrm{pairs}}}
\Bigr)
\\[4pt]
&\qquad\le\;
\min\{|\mathcal V|\,
C_4\,
\exp\!\bigl\{-C_1\bigl(|M(\bm r)|-t\bigr)\bigr\},1\}
\\[6pt]
&\qquad \le \min\{|\mathcal V|\,
C_4\,\exp\!\bigl\{-C_1\bigl(|M(\bm r)|-c/2\bigr)\bigr\},1\}
\qquad (t\le c/2)
\\[6pt]
&\qquad=\;\min\Big\{
\exp\!\Bigl\{
\log|\mathcal V|
+\log C_4
+\tfrac{C_1 c}{2}
-2C_1K_{\mathrm{pairs}}\,g(\bm r)
\Bigr\},1\Big\}
\\[6pt]
&\qquad=\;
\min\Big\{\exp\!\Bigl\{
-2C_1K_{\mathrm{pairs}}
\Bigl(
g(\bm r)
-\tfrac{\frac{c}{2}+\log C_4+\log|\mathcal V|}{2C_1K_{\mathrm{pairs}}}
\Bigr)
\Bigr\},1\Big\}.
\end{align}
By definition,
\[
\bigl|\mathcal C_{\mathcal{V}}(\mathcal{D})\bigr|
=
\sum_{\bm R\in S_K}
\mathbf 1\Bigl\{
\exists\, b\le|\mathcal V|:\ 
S(\bm\theta^{\ast(b)})=\bm R,\ 
g(\bm R)<\tfrac{c}{4K_{\mathrm{pairs}}}
\Bigr\}.
\]
Taking expectations and using linearity,
\[
\mathbb{E}_{\bm U,\mathcal V}\bigl|\mathcal C_{\mathcal{V}}(\mathcal{D})\bigr|
=
\sum_{\bm R\in S_K}
\mathbb{P}_{\bm U,\mathcal V}
\Bigl(
\exists\, b\le|\mathcal V|:\ 
S(\bm\theta^{\ast(b)})=\bm R,\ 
g(\bm R)<\tfrac{c}{4K_{\mathrm{pairs}}}
\Bigr).
\]
Applying \eqref{final bound on prob} to each $\bm R$ gives
\[
\mathbb{E}_{\bm U,\mathcal V}\bigl|\mathcal C_{\mathcal{V}}(\mathcal{D})\bigr|
\;\le\;
\sum_{\bm R\in S_K}
\min\Bigl\{
\exp\bigl(-C_5\bigl(g(\bm R)-\tilde g\bigr)\bigr),\,1
\Bigr\}.
\]
Finally, split the sum according to whether $g(\bm R)\ge\tilde g$ or
$g(\bm R)<\tilde g$:
\[
\begin{aligned}
\mathbb{E}_{\bm U,\mathcal V}\bigl|\mathcal C_{\mathcal{V}}(\mathcal{D})\bigr|
&\le
\sum_{\bm R:\,g(\bm R)\ge\tilde g}
\exp\bigl(-C_5\bigl(g(\bm R)-\tilde g\bigr)\bigr)
\;+\;
\sum_{\bm R:\,g(\bm R)<\tilde g} 1
\\[4pt]
&=
\sum_{\bm R:\,g(\bm R)\ge\tilde g}
\exp\bigl(-C_5\bigl(g(\bm R)-\tilde g\bigr)\bigr)
\;+\;
\bigl|\{\bm R:\ g(\bm R)<\tilde g\}\bigr|.
\end{aligned}
\]
\[
\text{where }
\tilde g
=
\frac{
\frac{c}{2}
+\log C_4
+\log|\mathcal V|
}{
C_5
},
\qquad
C_5
= 
2\,\frac{\Delta_{\min}^{(0)\,2}}{2 w_0 \bar\tau_n^2}
K_{\mathrm{pairs}}
=
\frac{\Delta_{\min}^{(0)\,2}}{w_0 \bar\tau_n^2}\,
K_{\mathrm{pairs}}\]\end{proof}
\end{document}